\documentclass{style/vldb}

\usepackage{amssymb}
\usepackage{amsmath}
\usepackage{color}
\usepackage{times}
\usepackage{cite}
\usepackage{enumitem}
\usepackage{arydshln}
\usepackage{url}
\usepackage{algorithm}
\usepackage{algorithmic}
\usepackage{flushend}
\usepackage{graphicx}

\usepackage{subfigure}
\usepackage{enumitem}

\usepackage[font={small,bf}]{caption} 

\renewcommand{\footnotesize}{\scriptsize}
\newcommand{\hindent}[1][1]{\hspace{#1\algorithmicindent}}

\title{Aggregate Estimations over Location Based Services}
\numberofauthors{1}
\author{ Weimo Liu$^\dag$, Md Farhadur Rahman$^\ddag$, Saravanan Thirumuruganathan$^\ddag$, Nan Zhang$^\dag$, Gautam Das$^\ddag$\\
\affaddr{The George Washington University$^\dag$, University of Texas at Arlington$^\ddag$}
}

\begin{document}
\maketitle

\begin{abstract}
{\em Location based services} (LBS) have become very popular in recent years. They range from map services (e.g., Google Maps) that store geographic locations of points of interests, to online social networks (e.g., WeChat, Sina Weibo, FourSquare) that leverage user geographic locations to enable various recommendation functions.
The public query interfaces of these services may be abstractly modeled as a $k$NN interface
over a database of two dimensional points on a plane:  given an arbitrary query point, the system returns the $k$  points in the database that are nearest to
the query point.
 In this paper we consider the problem of obtaining {\em approximate estimates} of SUM and COUNT aggregates by only
querying such databases via their restrictive public interfaces.
We distinguish between interfaces that return location information of the returned tuples (e.g., Google Maps), and interfaces that do not return location information (e.g.,  Sina Weibo).
For both types of interfaces, we develop  aggregate estimation algorithms that are based on novel techniques for precisely computing or approximately estimating the Voronoi cell of tuples. We discuss a comprehensive set of real-world experiments for testing our algorithms, including experiments on  Google Maps, WeChat, and Sina Weibo. 
\end{abstract}

\section{Introduction} \label{sec:intro}

\subsection{LBS with a $k$NN Interface}
{\em Location based services} (LBS) have become very popular in recent years. They range from map services (e.g., Google Maps) that store geographic locations of points of interests (POIs), to online social networks (e.g., WeChat, Sina Weibo, FourSquare) that leverage user geographic locations to enable various recommendation functions. The underlying data model of these services
may be viewed as a database of tuples that are either POIs (in case of map services) or users (in case of social networks), along
with their geographical coordinates (e.g., latitude and longitude) on a plane.

However, third-party applications and/or end users do not have complete and direct access to this entire database. 
The database is essentially ``hidden'', and access
is typically limited to a restricted public web query interface or API by which one can specify an arbitrary location as a query,
which returns at most $k$ nearest tuples to the query point (where $k$ is typically a small
number such as 10 or 50). For example, in Google maps it is possible to  specify an arbitrary location and
get the ten nearest Starbucks.
Thus, the query interfaces of these services may be abstractly modeled as a ``nearest neighbor'' $k$NN interface
over a database of two dimensional points on a plane:  given an arbitrary query point, the system returns the $k$  points in the database that are nearest to
the query point.

In addition, there are important differences among the services based on the type of information that is returned along with the $k$ tuples.
 Some services (e.g., Google maps) return the locations (i.e., the $x$ and $y$ coordinates) of the $k$ returned tuples.
 We refer to such services as  {\em Location-Returned LBS} (LR-LBS).
Other services (e.g., WeChat, Sina Weibo) return a ranked list of $k$ nearest tuples, but suppress the location of each tuple, returning only the tuple ID and perhaps some
other attributes (such as tuple name).
We refer to such services as {\em Location-Not-Returned LBS} (LNR-LBS).

Both types of services impose additional querying limitations, the most important being a per user/IP limit on the number of
queries one can issue over a given time frame (e.g., by default, Google map API imposes a query rate limit of 10,000 per user per day).

\subsection{Aggregate Estimations}
For many interesting third-party applications, it is important to collect {\em aggregate statistics} over the tuples contained in such hidden databases $-$ such as {\em sum}, {\em count}, or {\em distributions} of the
tuples satisfying certain selection conditions. For example, a hotel recommendation application would like to know the average review scores for Marriott vs Hilton hotels in Google Maps;
a cafe chain startup would like to know the number of Starbucks restaurants in a certain geographical region; a demographics researcher may wish to know the gender ratio of users of social networks in China etc.

Of course, such aggregate information can be obtained by entering into data sharing agreements with the location-based service providers, but this approach can often be extremely expensive,
and sometimes impossible if the data owners are unwilling to share their data. Therefore, in this paper we consider the problem of obtaining {\em approximate estimates} of such aggregates by only
querying the database via its restrictive public interface. Our goal is to minimize {\em query cost} (i.e., ask as few queries as possible) in an effort to adhere to the rate limits imposed by the interface, and yet make the aggregate estimations as accurate as possible.

The closest prior work is \cite{DKA+11}. This approach is based on generating random point queries,  estimating the area of the {\em Voronoi cell} \cite{de2000computational} of the returned top-1 tuple for each query, and estimating the aggregate from these top-1 tuples by making corrections for sampling bias using the area of the Voronoi cell.  However, there are several limitations of this work.
First, this approach works only for LR-LBS, but does not work  for LNR-LBS, and is thus inapplicable over a large variety of location based services such as  WeChat and Sina Weibo that do not return
precise location or distance information. Second, the approximate nature of the technique used for estimating the area of a Voronoi cell  makes the
overall aggregate estimation {\em inherently biased}. Third, the
method only uses the top-1 returned tuple for each query in its calculations  (the remaining $k-1$ tuples are ignored) thus leading to inefficiency in the estimation procedure. We discuss this
and other related work in \S\ref{sec:relWork}.

\subsection{Outline of Technical Results}

\smallskip\noindent {\bf Results over LR-LBS Interfaces:}
We first describe our results over LR-LBS interfaces. Like  \cite{DKA+11}, our approach is also based on generating random point queries and computing the area of Voronoi cells of the returned tuples, but a key differentiator is that our approach provides an efficient yet {\em precise} computation of the area of  Voronoi cells. This procedure is fundamentally
different from the approximate procedure used  in \cite{DKA+11} for estimating the area of Voronoi cells, and is one of the significant contributions of our paper.
This leads to {\em unbiased estimations} of SUM and COUNT aggregates
over LR-LBS interfaces; in contrast, the estimations in \cite{DKA+11} have inherent (and unknown) sampling bias.

Moreover, we also leverage the top-$k$ returned tuples of a query (and not just the top-1) by generalizing to the concept of a {\em top-$k$ Voronoi cell}, leading to significantly more efficient aggregate estimation algorithms.
We also developed four different techniques for reducing the estimation error (and thereby estimation error) over LR-LBS interfaces: faster initialization, leveraging history, variance reduction through dynamic selection of query results, and a Monte Carlo method which leverages current knowledge of upper/lower bounds on the Voronoi cell without sacrificing the unbiasedness of estimations.

We combine the above ideas to produce Algorithm LR-LBS-AGG, a completely unbiased estimator for COUNT and SUM queries with or without selection conditions. We note that AVG queries can be computed as SUM/COUNT.

\smallskip\noindent {\bf Results over LNR-LBS Interfaces:}
We also consider the  problem of aggregate estimations over LNR-LBS interfaces. To the best of our knowledge, this is a novel problem with no prior work.
Recall that such type of $k$NN interfaces only return a ranked list of the top-$k$ tuples in response
to a point query, and location information for these tuples is suppressed. None of the prior work for  LR-LBS interfaces can be extended to LNR-LBS
interfaces. For such interfaces, we develop aggregate estimation algorithms that are not completely bias-free, but can guarantee an arbitrarily small sampling bias.
The key idea here is the inference of a tuple's Voronoi cell to an arbitrary precision level with a small number of queries from {\em merely the ranks of the returned tuples}.

On a related note, we also show how one can infer the position of a tuple in LNR-LBS, again at a level of arbitrary precision - a problem, while of independent interest, is also critical for enabling the estimations of aggregates that feature selection conditions on tuples' locations (e.g., the COUNT of social network users within 10 meters of major highways).
We also study a subtle extension to cases where $k > 1$; in particular we study the challenge brought by this case by the (possibly) concave nature of top-$k$ Voronoi cells, and develop an efficient algorithm to detect potential concaveness and guarantee the accuracy of the inferred Voronoi cell.

We combine the above ideas to produce Algorithm LNR-LBS-AGG, an estimator for COUNT and SUM queries with or without selection conditions. Unlike Algorithm LR-LBS-AGG, this estimator
may be biased, but the bias can be controlled to any arbitrary desired precision. As before, we note that AVG queries can be computed as SUM/COUNT.

\subsection{Summary of Contributions} 

\begin{itemize}[noitemsep,topsep=1pt,parsep=1pt,partopsep=0pt]

\item Location based services have become very popular in recent years, and aggregate estimation over such ``hidden'' databases with their restricted  $k$NN query interfaces is an important problem
with numerous applications. In our work, we consider both LR-LBS (locations returned) as well as the more novel LNR-LBS (locations not returned) interfaces.
\item For  LR-LBS interfaces, we develop Algorithm LR-LBS-AGG for estimating COUNT and SUM aggregates with or without selection conditions. It represents a significant improvement over prior work along multiple dimensions: a novel way of precisely calculating Voronoi cells lead to completely unbiased estimations; top-$k$ returned tuples are leveraged rather than only top-1; several innovative techniques developed for reducing error and increasing efficiency.
\item For  LNR-LBS interfaces, we develop Algorithm LNR-LBS-AGG for estimating COUNT and SUM aggregates with or without selection conditions.This is a novel problem with no prior work. The estimated
aggregates are not bias-free, but the sampling bias can be controlled to any desired precision. Among several key ideas, we show how a Voronoi cell can be inferred to an arbitrary degree of precision from merely the ranks of returned tuples to point queries.
\item Our contributions also include a comprehensive set of real-world experiments. Specifically, we conducted online tests over a number of real-world LBS, e.g., Google Maps (LR-LBS) for estimating the number of Starbucks in US (and compared the results with the ground truth published by Starbucks); WeChat and Sina Weibo for estimating the percentage of male/female users in China.
\end{itemize}

\section{Background} \label{sec:pre}

\subsection{Model of LBS}\label{subsec:lbsModels}

A location based service (LBS) supports $k$NN queries over a database $D$ of tuples with location information. These tuples can be points of interest (e.g. Google Maps) or users (e.g. WeChat, Sina Weibo). A $k$NN query $q$ takes as input a location (e.g., longitude/latitude combination), and returns the top-$k$ nearest tuples selected and ranked according to a pre-determined ranking function. Since the only input to a query is a location, we use $q$ to also denote the query's location without introducing ambiguity. Most of the popular LBS follow $k$NN query model. For most parts of the paper, we consider the ranking function to be Euclidean distance between the query location and each tuple's location. Extensions to other ranking functions are discussed in \S\ref{subsec:lbsConstraints}.

Note that tuples in an LBS system contain not only location information but other many other attributes - e.g.,  a POI in Google Maps includes attributes such as POI name, average review ratings etc. Depending on which attributes of a tuple are returned by the $k$NN interface - more specifically, whether the location of a tuple is returned - we can classify LBS into two main categories:

\vspace{1mm}
\noindent{\bf LR-LBS:} A Location-Returned-LBS (LR-LBS) returns the precise location for each of the top-$k$ returned tuples, along with possibly other attributes. Google Maps, Bing Maps, Yahoo!~Maps, etc., are prominent examples of LR-LBS, as all of them display the precise location of each returned POI. Note that some LBS may return a variant of the precise locations - e.g., Skout and Momo returns not the precise location of a tuple, but the precise distance between the query location and the tuple location. We consider such LBS to be in the LR-LBS category because, through previously studied techniques such as trilateration (e.g., \cite{LZG+13}), one can infer the precise location of a tuple with just 3 queries.

\vspace{1mm}
\noindent{\bf LNR-LBS:} A Location-Not-Returned-LBS (LNR-LBS), on the other hand, does {\em not} return tuple locations - i.e., only other attributes such as name, review rating, etc., are returned. This category is more prevalent among location based social networks (presumably because of potential privacy concerns on precise user locations). Examples here include WeChat, which returns attributes such as name, gender, etc., for each of the top-$k$ users, but not the precise location/distance. Other examples include Sina Weibo, WeChat, etc., which feature a similar query return semantics.

\vspace{1mm}
\noindent{\bf Common Interface Features and Limitations:} Generally speaking, there are two ways through which an LBS (either LR- or LNR-LBS) supports a $k$NN query. One is an interactive web or API interface which allows a location to be explicitly specified as a latitude/longitude pair. Google Maps is an example to this end. Another common way is for the LBS (e.g., as a mobile app) to directly retrieve the query location from a positioning service (such as GPS, WiFi or Cell towers) and automatically issue a $k$NN query accordingly. In the second case, there is no explicit mechanism to enter the location information. Nonetheless, it is important to note that, even in this case, we have the ability to issue a query from any arbitrary location {\em without} having to physically travel to that location. All mobile OS have debugging features that allow arbitrary location to be used as the output of the positioning (e.g., GPS) service. 

Many LBS also impose certain interface restrictions: One is the aforementioned top-$k$ restriction (i.e., only the $k$ nearest tuples are returned). Another common one is a {\em query rate limit} - i.e., many LBS limit the maximum number of $k$NN queries one can issue per unit of time. For example, by default Google Maps allows 10,000 location queries per day while Sina Weibo allows only 150 queries per hour. This is an important constraint for our purpose because it makes the {\em query-issuing} process the bottleneck for aggregate estimation. To understand why, note that even with the generous limit provided by Google Maps, one can issue only 7 queries per minute - this 8.6 second per query overhead\footnote{Of course, one can shorten it with multiple IP addresses and API accounts - but similarly, one can use parallel processing to speed up offline processing as well.} is orders of magnitude higher than any offline processing overhead that may be required by the aggregate estimation algorithm. Thus, this interface limitation essentially makes {\em query cost} the No.~1 performance metric to optimize for aggregate estimation.
An LBS might impose other, more subtle constraints - e.g., a maximum coverage limit which forbids tuples far away (say more than 5 miles away) from a query location to be returned. We shall discuss about these constraints in \S\ref{subsec:lbsConstraints}.

\subsection{Voronoi Cells}
{\em Voronoi cell} \cite{de2000computational} is a key geometry concept used extensively by our algorithms developed in the paper. Thus, we introduce this concept here as part of the preliminaries. Consider each tuple $t \in D$ as a point on a Euclidean plane bounded by a box $B$ (which covers all tuples in $D$). We have the following definition.

\newtheorem{definition}{Definition}
\begin{definition} [Voronoi Cell] Given a tuple $t \in D$, the Voronoi cell of $t$, denoted by $V(t)$, is the set of points on the $B$-bounded plane that are closer to $t$ than any other tuple in $D$.
\end{definition}

Note that the $B$-bound ensures that each Voronoi cell is a finite region. The Voronoi cells of different tuples are mutually exclusive - i.e., the {\em Voronoi diagram} is the subdivision of the plane into regions, each corresponding to all query locations that would return a certain tuple as the nearest neighbor\footnote{We assume general positioning\cite{de2000computational} - i.e., no two tuples have the exact same location and no four points on the same circle.}.

For the purposes of our paper, we define an extension of the Voronoi cell concept to accommodate the top-$k$ (when $k > 1$) query return semantics. Specifically, given a tuple $t \in D$, we define the {\em top-$k$ Voronoi cell} of $t$, denoted by $V_k(t)$,
as the set of query locations on the plane that return $t$ as one of the top-$k$ results. 
There are four important observations about this concept:

First, the top-$k$ Voronoi cells for different tuples are no longer mutually exclusive. Each location $l$ belongs to exactly $k$ top-$k$ Voronoi cells corresponding to the top-$k$ tuples returned by query over $l$.
Second, our concept of top-$k$ Voronoi cells is {\em fundamentally different} from the $k$-th ordered Voronoi cells in geometry \cite{de2000computational} - each of which is formed by points with the exact same $k$ closest tuples. The difference is illustrated in Figure~\ref{fig:topKVoronoiDiagrams} - while each colored region is a $k$-th ordered Voronoi cell, a top-$k$ Voronoi cell may cover multiple regions with different colors. For example, the top-2 Voronoi cell for tuple A is marked by a red border and consists of two different $k$-th ordered Voronoi cells (AB and AE).

\begin{figure}[htp]
    \centering
    \subfigure{\includegraphics[height=25mm, width=30mm]{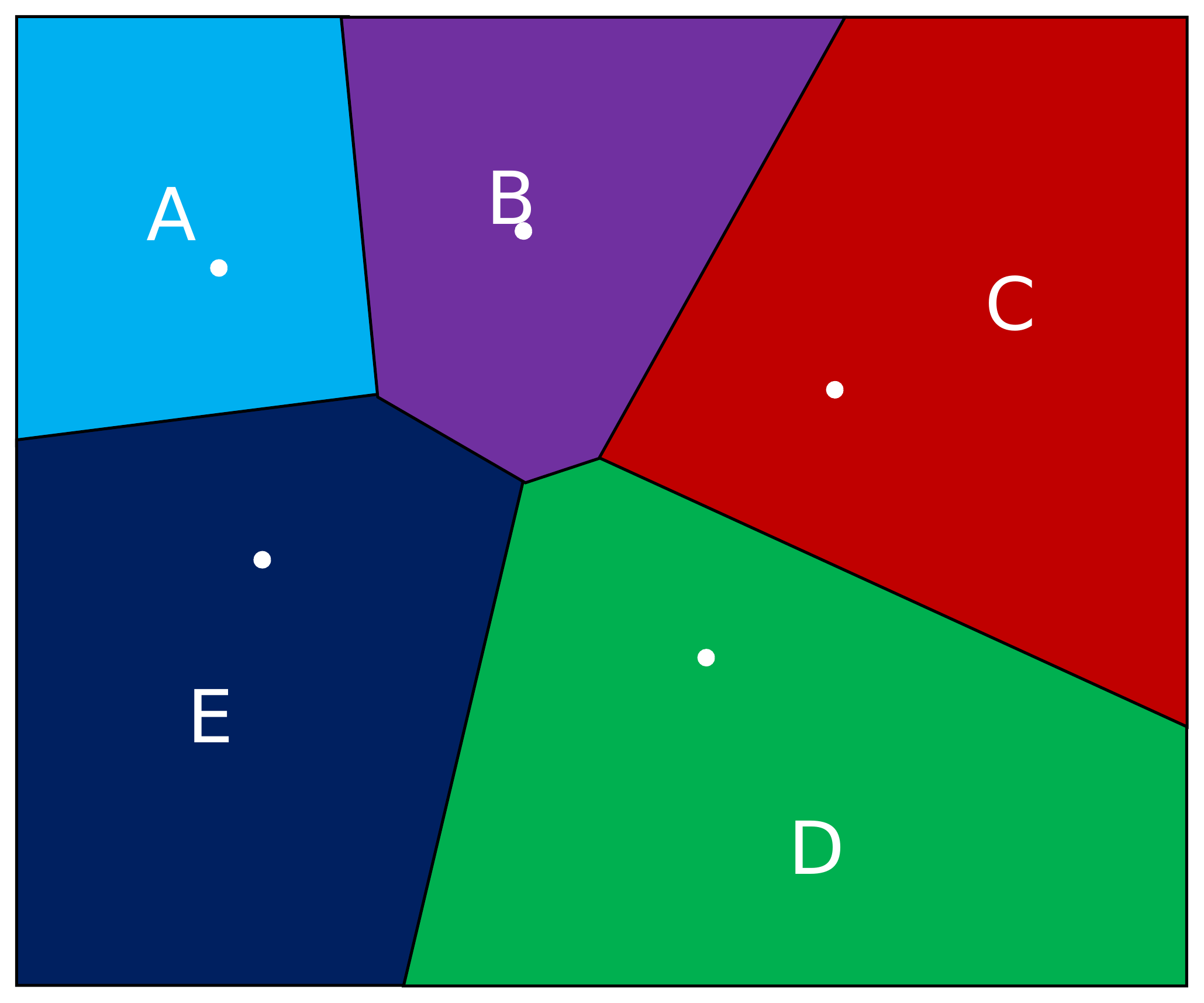}}\quad
    \subfigure{\includegraphics[height=25mm, width=30mm]{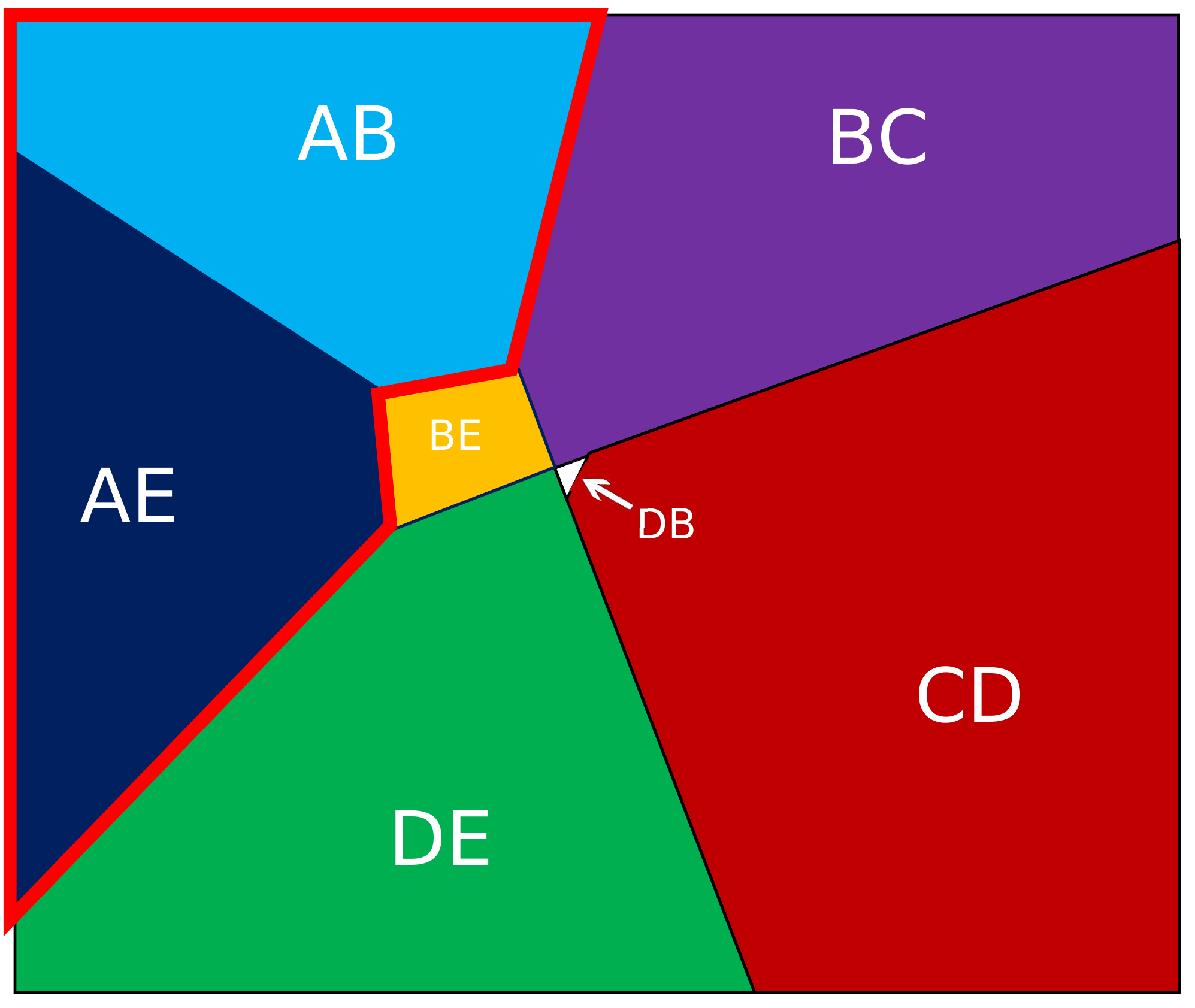}}
    \caption{Concavity of top-$k$ Voronoi Diagrams}
    \label{fig:topKVoronoiDiagrams}
\end{figure}

Third, while both top-1 Voronoi cells and $k$-th order Voronoi cells are guaranteed to be convex \cite{de2000computational}, the same does not hold for top-$k$ Voronoi cells when $k > 1$. For example, from Figure~\ref{fig:topKVoronoiDiagrams} we can see that the aforementioned top-$2$ Voronoi cell for tuple A is concave.
Fourth, a top-$k$ Voronoi cell tend to contain many more edges than a top-1 Voronoi cell. As we shall discuss later in the paper, the larger number of edges and the potential concaveness makes computing the top-$k$ Voronoi cell of a tuple $t$ more difficult.

\begin{figure*}[ht]                                                                                 
\begin{minipage}[t]{0.4\linewidth}                                                                  
\centering                                                                                          
    \includegraphics[width = 65mm, height = 35mm]{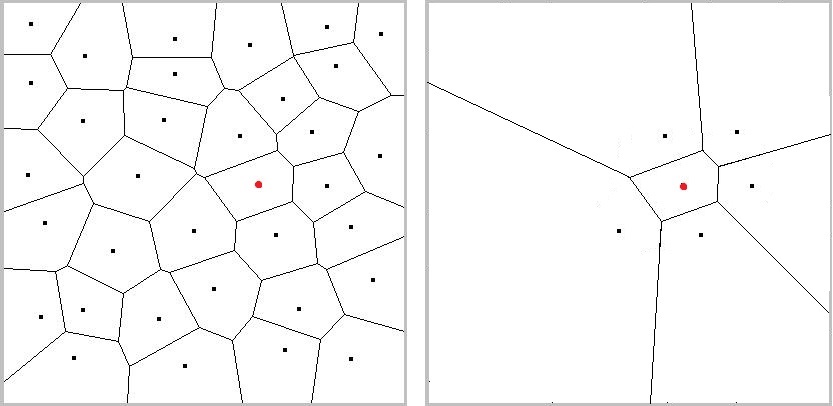}                                 
    \vspace{-3mm}\caption{Illustration of Theorem 1}                                                
    \label{fig:thm1_illustration}                                                                   
\end{minipage}                                                                                      
\hspace{1mm}                                                                                        
\begin{minipage}[t]{0.3\linewidth}                                                                  
\centering                                                                                          
    \includegraphics[width = 50mm, height = 36mm]{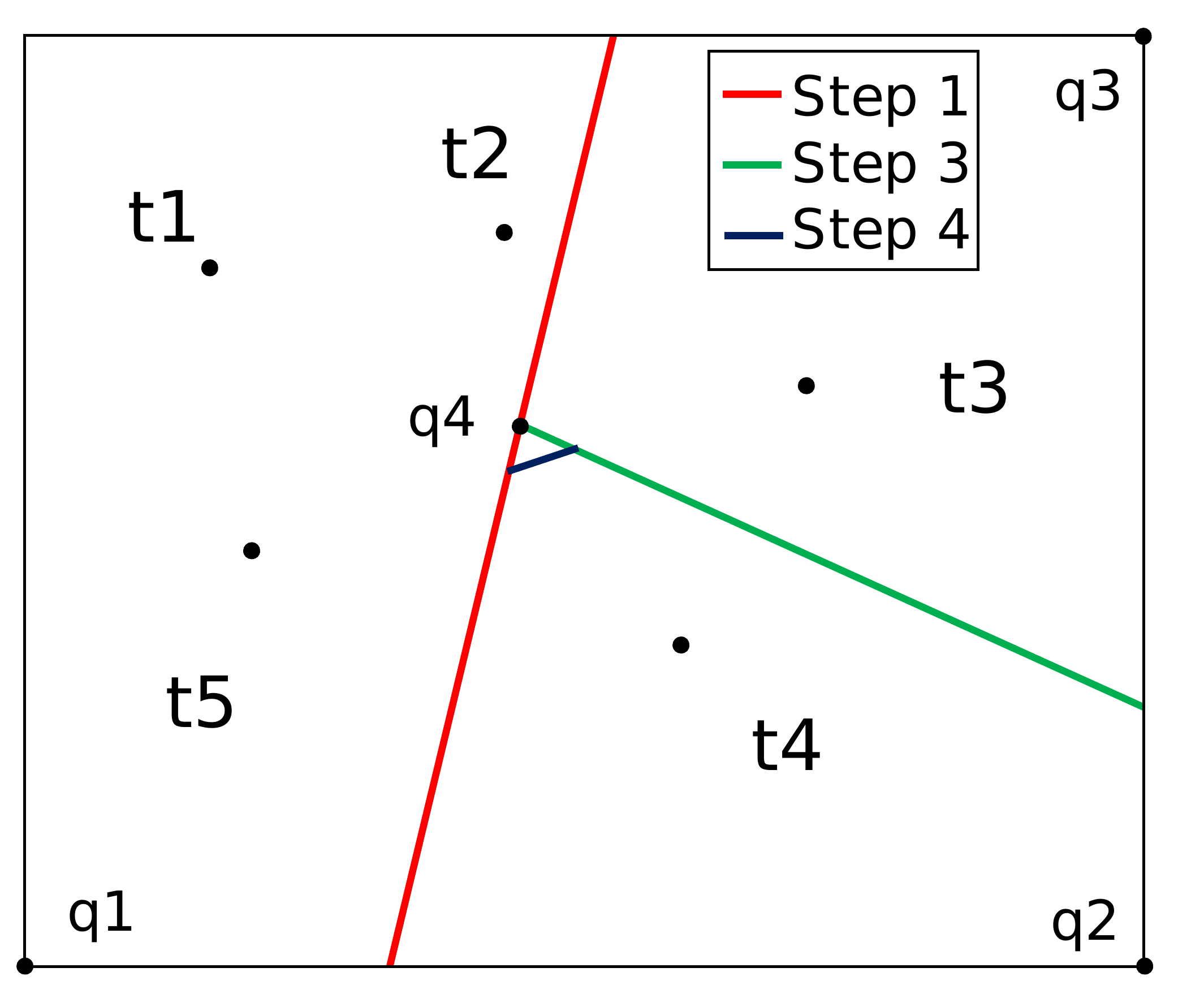}                               
    \vspace{-3mm}\caption{Illustration of LR-LBS-AGG}                                               
    \label{fig:lr_lbs_agg_illstration}                                                              
\end{minipage}                                                                                      
\hspace{1mm}                                                                                        
\begin{minipage}[t]{0.3\linewidth}                                                                  
\centering                                                                                          
    \includegraphics[width = 50mm, height = 36mm]{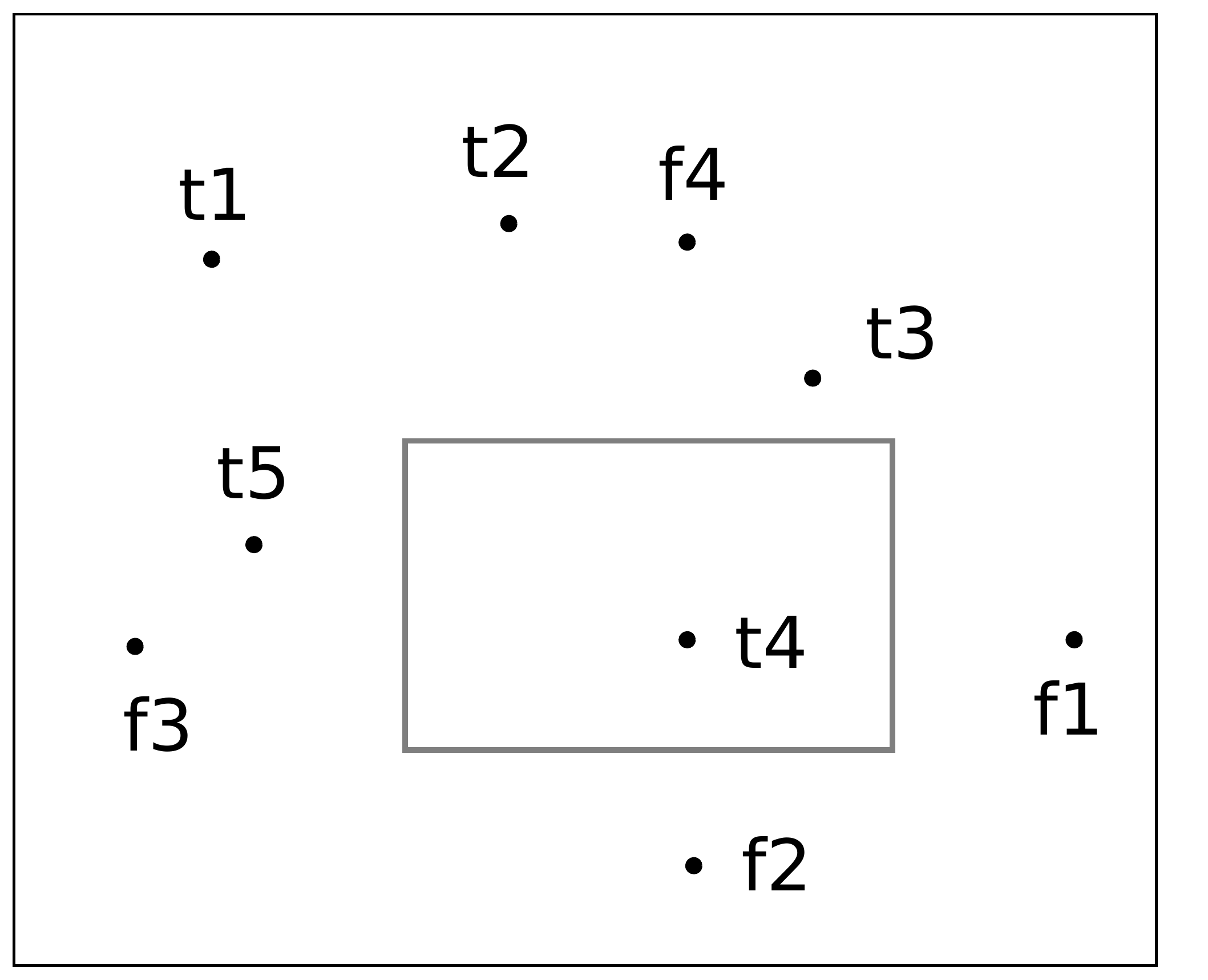}                       
    \vspace{-3mm}\caption{Faster Initialization - Success}                                          
    \label{fig:fast_init_succ}                                                                      
\end{minipage}                                                                                      
\hspace{1mm}                                                                                        
\end{figure*}           

\subsection{Problem Definition}
In this paper, we address the problem of aggregate estimations over LBS. Specifically, we consider aggregate queries of the form
\textsc{Select Aggr($t$) From $D$ Where} {\em Cond}
\noindent where \textsc{Aggr} is an aggregate function such as SUM, COUNT and AVG that can be evaluated over one or more attributes of $t$, and {\em Cond} is the selection condition. Examples include the COUNT of users in WeChat or AVG rating of restaurants in Texas at Google Maps.

There are two important notes regarding the selection condition {\em Cond}. First, we support any selection condition that can be independently evaluated over a single tuple - i.e., it is possible to determine whether a tuple $t$ satisfies {\em Cond} based on nothing but $t$. Second, for both LR- and LNR-LBS, we support the specification of a tuple's location as part of {\em Cond} - even when such a location is not returned, like in LNR-LBS. This is possible thanks to what we shall discuss in \S\ref{subsec:tuplePositionComputation} - i.e., even with LNR-LBS, one can derive the location of a tuple to arbitrary precision after issuing a small number of queries. As such, we support aggregates such as the percentage of female WeChat users in Washington, DC).

In most part of the technical sections, we focus on aggregates without selection conditions - the straightforward extensions to various types of selection conditions will be discussed in \S\ref{sec:ext}.

\vspace{1mm}
\noindent {\bf Performance Measures:} The performance of an aggregate estimation algorithm is measured in terms of efficiency and accuracy.
Given the query-rate limit enforced by all LBS, the efficiency is measured by {\em query cost} -
i.e. the number of queries and/or API calls that the algorithm issues to LBS.
Often, we are given a fixed budget (based on the rate limits) and hence designing an efficient algorithm 
that generates accurate estimates within the budgetary constraints is crucial.
The accuracy of an estimation $\tilde{\theta}$ of an aggregate $\theta$ could be measured by the standard measure of
{\em relative error } $|\tilde{\theta} - \theta|/\theta$.
Note that, for any sampling-based approach (like ours), the relative error is determined by two factors: {\em bias}, i.e. 
$|E(\tilde{\theta} - \theta)|$
, and {\em variance} of $\tilde{\theta}$.
The mean squared error, MSE of the estimation is computed as MSE = bias$^2$ + variance.

An interesting question often arises in practice is how we can determine the relative error achieved by our estimation. If the population variance is known, then one can apply standard statistics techniques to compute the confidence interval of aggregate estimations\cite{freedman2009statistical}. Absence of such knowledge, a common practice is to approximate the population variance with {\em sample} variance, which can be computed from the samples we use to generate the final estimation and use Bessel's correction \cite{freedman2009statistical} to correct the result.

\section{LR-LBS-AGG}\label{sec:lrlbsagg}

In this section, we develop LR-LBS-AGG, our algorithm for generating unbiased SUM and COUNT estimations over an LR-LBS query interface. Specifically, we start with introducing our key idea of precisely computing the (top-$k$) Voronoi cell of a given tuple, which enables the unbiased aggregate estimations. While this idea guarantees unbiasedness, it may require a large number of queries per (randomized) estimation, leading to a large estimation variance (and therefore, error) when the query budget is limited. Hence we develop four techniques for reducing the estimation error while {\em maintaining} the complete unbiasedness of aggregate estimations. Finally, we combine all ideas to produce Algorithm LR-LBS-AGG at the end of this section.

\subsection{Key Idea: Precisely Compute Voronoi Cells} \label{sec:lrk}

\noindent{\bf Reduction to Computing Voronoi Cells:} We start by describing a baseline design which illustrates why the problem of aggregate estimations over an LBS's $k$NN interface ultimately boils down to computing the volume of the Voronoi cell corresponding to a tuple $t$. As an example, consider the estimation of COUNT(*) (over a given region) through an LR-LBS with a top-1 interface.

We start by choosing a location $q$ uniformly at random from the region, and then issue a query at $q$. Let $t$ be the tuple returned by $q$. Suppose that we can compute the Voronoi cell of $t$ (as defined in \S\ref{sec:pre}), say $V(t)$. A key observation here is that the sampling probability of $t$, i.e., the probability for the above-described randomized process to return $t$, is exactly
$p(t) = \frac{|V(t)|}{|V_0|}$
where $|V(t)|$ and $|V_0|$ are the volume of $V(t)$ and the entire region, respectively. Note that knowledge of $p(t)$ directly leads to a completely unbiased estimation of COUNT(*): $r = 1/p(t)$, because
\begin{align}
\mathrm{Exp}(r) = \sum_{t \in D} p(t) \cdot \frac{1}{p(t)} = |D|, \label{equ:uce}
\end{align}
where $\mathrm{Exp}(\cdot)$ is the expected value of the estimation (taken over the randomness of the estimation process), and $|D|$ is the total number of tuples in the database. From (\ref{equ:uce}), one can see that every SUM and COUNT aggregate we support can be estimated without bias - the only change required is on the numerator of estimation. Instead of having 1 as in the COUNT(*) case, it should be the evaluation of the aggregate over $t$ - e.g., if we need to estimation SUM($A_1$) where $A_1$ is an attribute, then the numerator should be $t[A_1]$, i.e., the value of $A_1$ for $t$. If the aggregate is COUNT with a selection condition, then the numerator should be either 1 if $t$ satisfies the condition, or 0 if it does not.
One can see from the above discussions that, essentially, the problem of enabling unbiased SUM and COUNT estimations is reduced to that of {\em precisely} computing the volume of $V(t)$, i.e., the Voronoi cell of a given tuple $t$.

\vspace{2mm}
\noindent{\bf Computing Voronoi Cells:} For computing the Voronoi cell of a given tuple, a nice feature of the LR-LBS interface is that it returns the precise location of every returned tuple. Clearly, if we can somehow ``collect'' all tuples with Voronoi cells adjacent to that of $t$, then we can precisely compute the Voronoi cell of $t$ based on the locations of these tuples (and $t$). As such, the key challenges here become: (1) how do we collect these tuples and (2) how do we know if/when we have collected all tuples with adjacent Voronoi cells to $t$? Both challenges are addressed by the following theorem which forms the foundation of design of Algorithm LR-LBS-AGG.

\newtheorem{theorem}{Theorem}
\begin{theorem} \label{thm:lr}
Given a tuple $t \in D$ and a subset of tuples $D^\prime \subseteq D$ such that $t \in D^\prime$, the Voronoi cell of $t$ defined according to $D^\prime$, represented by $P^\prime$, is the same as that according to the entire dataset $D$, denoted by $P$, if and only if for all vertices $v$ of $P^\prime$, all tuples returned by the nearest neighbor query issued at $v$ over $D$ belong to $D^\prime$.
\end{theorem}
\begin{proof}
First, note that there must be $P \subseteq P^\prime$, because for a given location $q$, if there is already a tuple $t^\prime$ in $D^\prime$ that is closer to $q$ than $t$, then there must at least one tuple in $D$ that is closer to $q$ than $t$. Second, if $P \neq P^\prime$ (i.e., $P \subset P^\prime$), then there must at least one vertex of $P^\prime$, say $v$, that falls outside $P$. i.e. there must exist a tuple $t_0 \in (D \backslash D^\prime)$ that is closer to $v$ than all tuples in $D^\prime$. 
\end{proof}

{\bf Example 1:}
Figure~\ref{fig:thm1_illustration} provides an illustration for Theorem~\ref{thm:lr}.
In order to compute the Voronoi cell of the tuple corresponding to the red dot, it suffices to know the location of the adjacent tuples.
Since each Voronoi edge is a perpendicular bisector between the adjacent tuples, 
the entire Voronoi cell can be computed as the convex shape induced by the intersections of the edges.

Theorem~\ref{thm:lr} answers both challenges outlined above: it tells us when we have collected all ``adjacent'' tuples - when all vertices of $t$'s Voronoi cell computed from the collected tuples return only collected tuples. It also tells us how to collect more ``adjacent'' tuples when not all of them have been collected - any vertex which fails the test naturally returns some tuples that have not been collected yet, adding to our collection and starting the next round of tests.

According to the theorem, a simple algorithm for constructing the exact Voronoi cell for $t$ is as follows: We start with $D^\prime = \{t\}$. Now the Voronoi cell is the entire region (say an extremely large rectangle). We issue queries corresponding to its four vertices. If any query returns a point we have not seen yet - i.e., not in $D^\prime$ - we append it to $D^\prime$, recompute the Voronoi cell, and repeat the process. Otherwise, if all queries corresponding to vertices of the Voronoi cell return points in $D^\prime$, we have obtained the real Voronoi cell for $t \in D$.
One can see that the query complexity of this algorithm is $O(n)$, where $n$ is the number of points in the database $D$, because each query issued either confirms a vertex of the final Voronoi cell (which has at most $n - 1$ vertices), or returns us a new point we have never seen before (there are at most $n - 1$ of these too). It is easy to see that the bound is tight - as one can always construct a Voronoi cell that has $n - 1$ edges and therefore requires $\Omega(n)$ top-1 queries to discover (after all, each such query returns only 1 tuple). An example here is when $t$ is in the center of a circle, on which the other $n - 1$ points are located.
Algorithm~\ref{alg:lrlbsagg_baseline} shows the pseudocode of the baseline approach which we improve in Section~\ref{sec:rer}.

\begin{algorithm}[!htb]
\caption{{\bf LNR-LBS-AGG-Baseline}}
\begin{algorithmic}[1]
\label{alg:lrlbsagg_baseline}
\STATE {\bf while} query budget is not exhausted
    \STATE \hindent $q$ = location chosen uniformly at random; \quad $t$ = query($q$)
    \STATE \hindent $V(t) = V_0$; \quad $D'=\{t\}$
    \STATE \hindent {\bf repeat} till $D'$ does not change between iterations
        \STATE \hindent[2] {\bf for} each vertex $v$ of $V(t)$: $D'=D'\cup$ query$(v)$
        \STATE \hindent[2] Update $V(t)$ from $D'$
\STATE Produce aggregate estimation using samples
\end{algorithmic}
\end{algorithm}

{\bf Example 2:}
Figure~\ref{fig:lr_lbs_agg_illstration} provides a simple run-through of the algorithm for a dataset with 5 tuples $\{t_1, \ldots, t_5\}$.
Suppose we wish to compute $V(t_4)$. Initially, we set $D'=\{t_4\}$ and $V(t_4)=V_0$, the entire bounding box.
We issue query $q_1$ that returns tuple $t_5$ and hence $D'=\{t_4, t_5\}$. 
We now obtain a new Voronoi edge that is the perpendicular bisector between $t_4$ and $t_5$.
The Voronoi cell after step 1 is highlighted in light grey.
In step 2, we issue query $q_2$ that returns $t_4$ resulting in no update.
In step 3, we issue query $q_3$ that returns $t_3$. 
$D'=\{t_3, t_4, t_5\}$ and we obtain a new Voronoi edge as the perpendicular bisector between $t_3$ and $t_4$ depicted in dark medium gray.
In step 4, we issue query $q_4$ that returns $t_2$ resulting in the final Voronoi edge depicted in dark grey.
Further queries over the vertices for $V(t_4)$ does not result in new tuples concluding the invocation of the algorithm.

\noindent{\bf Extension to $k > 1$:} Interestingly, no change is required to the above algorithm when we consider the top-$k$ Voronoi cell rather than the traditional, i.e., top-1 Voronoi cell. To understand why, note that Theorem~\ref{thm:lr} directly extends to top-$k$ Voronoi cells - as a top-$k$ Voronoi computed from $D^\prime$ still must completely cover that for $D$; and any vertex of the top-$k$ Voronoi from $D^\prime$ which is outside that from $D$ must return at least one tuple outside $D^\prime$. We further describe how to leverage $k>1$ in Sections~\ref{subsec:lrkgt1} and \ref{subsec:lnrkgt1}.

\subsection{Error Reduction} \label{sec:rer}

Before describing the various error reduction techniques we develop for aggregate estimations over LR-LBS, we would like to first note that, while we use the term ``error reduction'' as the title of this subsection, some of the techniques described below indeed focus on making the computation of a Voronoi cell more efficient. The reason why we call all of them ``error reduction'' is because of the inherent relationship between efficiency and estimation error - if the Voronoi-cell computation becomes more efficient, then we can do so for more samples, leading to a larger sample size and ultimately, a lower estimation error (which is inversely proportional to the square root of sample size\cite{freedman2009statistical}).

\begin{figure*}[t]                                                                                  
\begin{minipage}[t]{0.3\linewidth}                                                                  
\centering                                                                                          
    \includegraphics[width = 55mm, height = 32mm]{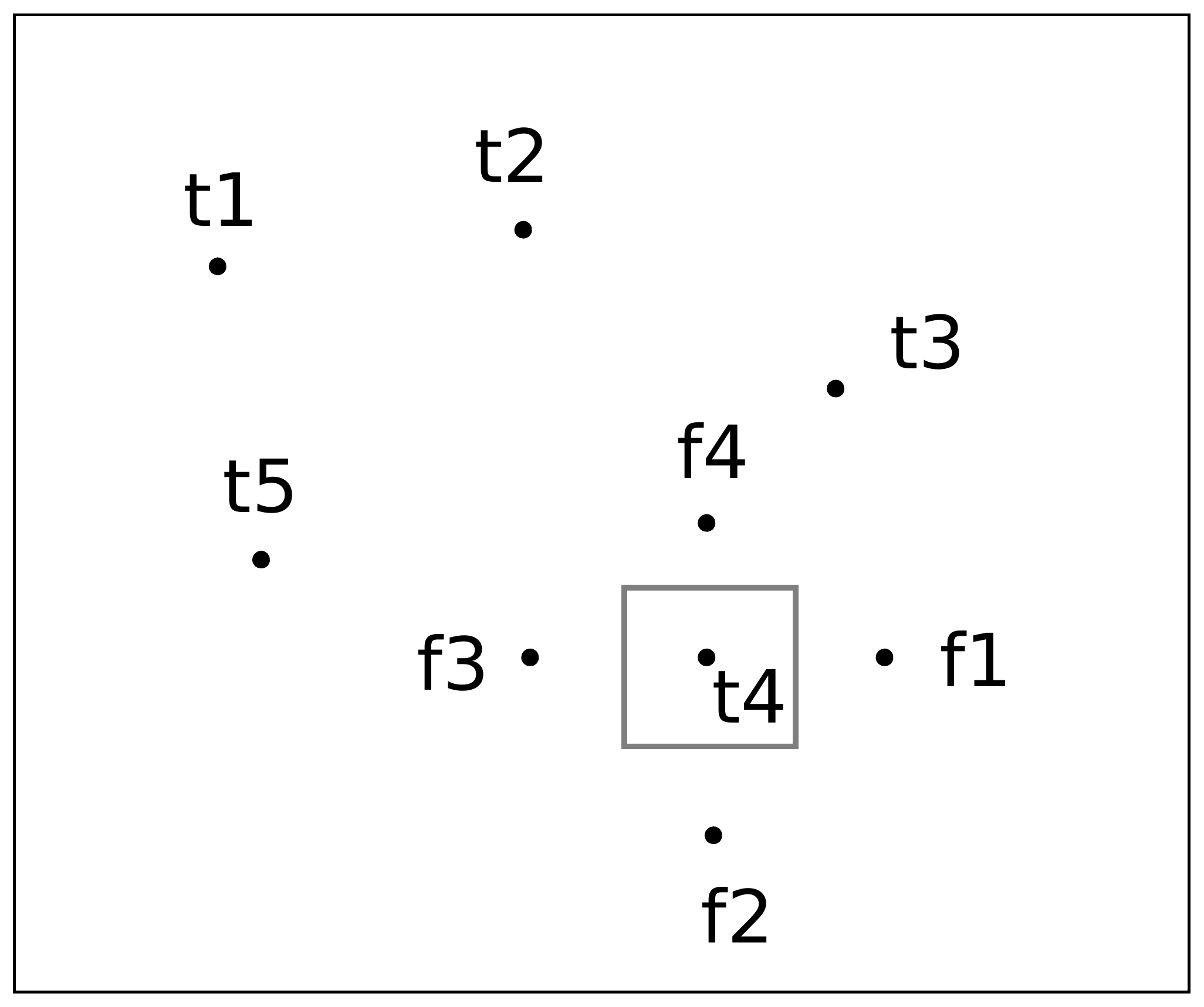}                       
    \vspace{-6mm}\caption{Faster Initialization - Failure}                                          
    \label{fig:fast_init_fail}                                                                      
\end{minipage}                                                                                      
\hspace{1mm}                                                                                        
\begin{minipage}[t]{0.3\linewidth}                                                                  
\centering                                                                                          
    \includegraphics[width = 55mm, height = 32mm]{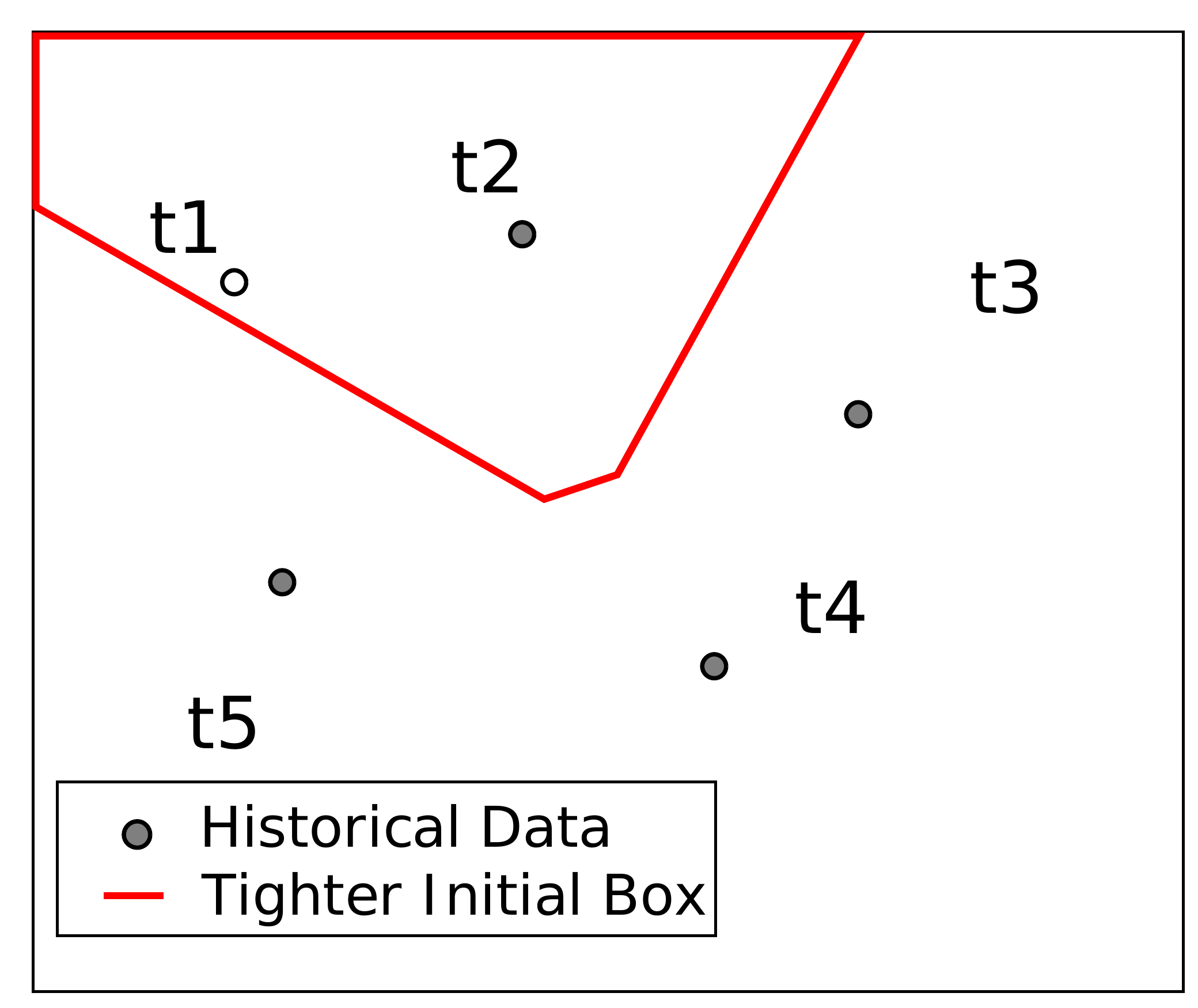}                     
    \vspace{-6mm}\caption{Leveraging History}                                                       
    \label{fig:leverage_history}                                                                    
\end{minipage}                                                                                      
\hspace{1mm}                                                                                        
\begin{minipage}[t]{0.3\linewidth}                                                                  
\centering                                                                                          
    \includegraphics[width = 55mm, height = 32mm]{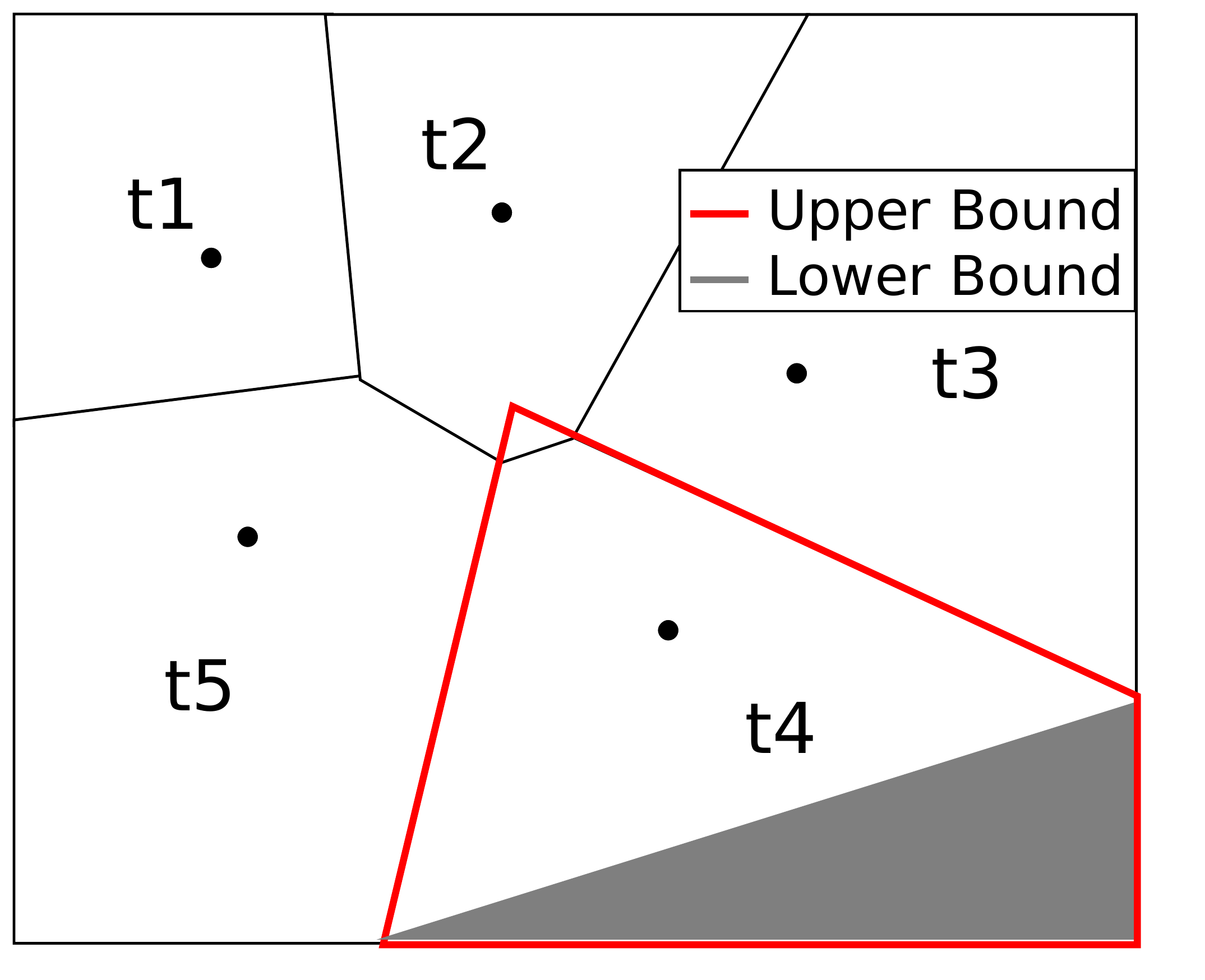}                      
    \vspace{-6mm}\caption{Upper/Lower Bounds}                                                       
    \label{fig:upper_lower_bounds}                                                                  
\end{minipage}                                                                                      
\hspace{-2mm}                                                                                       
\end{figure*}

\subsubsection{Faster Initialization} \label{sec:fip}

A key observation from the design in \S\ref{sec:lrk} is its bottleneck: the initialization process. At the beginning, we know nothing about the database other than (1) the location of tuple $t$, and (2) a large bounding box corresponding to the area of interest for the aggregate query. Naturally, $D^\prime = \{t\}$, leading to the initial Voronoi cell being the bounding box, and our first four queries being the corners of these bounding boxes. Of course, the tentative Voronoi cell will quickly close in to the real one with speed close to a binary search - i.e., the average-case query cost is at $\log$ scale of the bounding box size. Nonetheless, the initialization process can still be very costly, especially when the bounding box is large.

To address this problem, we develop a {\em faster initialization} technique which features a simple idea: Instead of starting with $D^\prime = \{t\}$, we insert four fake tuples into $D^\prime$, say $D^\prime = \{t, t^\mathrm{F}_1, \ldots, t^\mathrm{F}_4\}$, where $t^\mathrm{F}_1$, $\ldots$, $t^\mathrm{F}_4$ form a bounding box around $t$. The size of the bounding box should be conservatively large - even though a wrongly set size will not jeopardize the accuracy of our computation - as we shall show next.

By computing the initial Voronoi cell from $D^\prime$ and then issue queries corresponding to its vertices, there are two possible outcomes: One is that these queries return enough real tuples (besides $t$, of course) that, after excluding the fake ones from $D^\prime$, we still get a bounded Voronoi cell for $t$. One can see that, in this case, we can simply continue the computation while having saved a significant number of initialization queries. The other possible outcome, however, is when the bounding box is set too small, and we do not have enough real tuples to ``bound'' $t$ with a real Voronoi cell. Specifically, in the extreme-case scenario, all four vertices of the initial Voronoi cell could return $t$ itself. In this case, we simply revert back to the original design, wasting nothing but four queries.

One can see that the faster initialization process still guarantees the exact computation of a tuple's Voronoi cell. It has the potential to save a large amount of initialization queries in the average-case scenario, while in the worst case, it wastes at most four queries.
Algorithm~\ref{alg:fast_init} provides the pseudocode for faster initialization strategy.

\begin{algorithm}[!htb]
\caption{{\bf Fast-Init}}
\begin{algorithmic}[1]
\label{alg:fast_init}
\STATE {\bf Input:} $t$; \qquad {\bf Output:} $V(t)$
\STATE $D'=\{t, t_1^F, t_2^F, t_3^F, t_4^F\}$; Update $V(t)$ based on $D'$
\STATE {\bf If} all queries over vertices of $V(t)$ return $t$, {\bf then return} $V_0$
\STATE {\bf repeat} till $D'$ does not change between iterations
    \STATE \hindent {\bf for} each vertex $v$ of $V(t)$: $D'=D'\cup$ query$(v)$
    \STATE \hindent Update $V(t)$ from $D'$
\STATE {\bf return} $V(t)$
\end{algorithmic}
\end{algorithm}

{\bf Example 3:}
Figures~\ref{fig:fast_init_succ} and \ref{fig:fast_init_fail} show two different scenarios where the strategy is successful and not successful respectively
based on whether the bounding box due to fake tuples is conservatively large.
Given a small dataset with tuples $\{t_1, \ldots, t_5\}$, we initialize them with a bounding box corresponding to fake tuples $\{f_1, \ldots, f_4\}$.
In Figure~\ref{fig:fast_init_succ}, the initial bounding box is tight enough 
and results in the computation of the precise $V(t_4)$ with much lower query cost 
(i.e. only tuples $\{t_3, t_5\}$ are visited as against tuples $\{t_2, t_4, t_5\}$ for the example of Algorithm~\ref{alg:lrlbsagg_baseline}.
On the other hand, if the bounding box is not tight (as in Figure~\ref{fig:fast_init_fail}),
then queries over all the vertices of the bounding box return $t_4$.
We then revert back to the original bounding box $V_0$ that covers the entire region.

\subsubsection{Leverage history on Voronoi-cell computation} \label{sec:lhi}
Another natural optimization is to leverage the information that is gleaned from computing the Voronoi cells of past tuples
to compute a tighter initial Voronoi cell.
Recall that our algorithm to compute Voronoi cell of a tuple $t$ (i.e $V(t)$), using Theorem~\ref{thm:lr} starts with an initial Voronoi cell
that is an extremely large bounding box that covers the entire plane that then converges to $V(t)$.
In the process of computing this Voronoi cell, our algorithm retrieved additional new tuples (by issuing queries for each vertex of the bounding box).
Notice that for a LBS with static tuples (such as POIs in Google Maps), the results of location query ordered by distance remains static.
Hence it is not necessary to restart every iteration of the algorithm with the same large bounding box.
Specifically, when computing the Voronoi cell for the next tuple, we could leverage history by
starting with a ``tighter'' initial bounding box whose vertices are the set of tuples that we have seen so far.
In other words, we reuse the tuples that we have seen so far and make them as input to further rounds.
Notice that this approach remains the same for both $k=1$ and $k>1$.
Since the location of each tuple in top-$k$ are returned in LR-LBS, each of these tuples could be leveraged.
As we see more tuples, the initial Voronoi cell becomes more granular
resulting in substantial savings in query cost.
Algorithm~\ref{alg:leverage_history} provides the pseudocode for the strategy.
While the pseudocode uses the simple perpendicular bisector half plane approach\cite{de2000computational},
it could also use more sophisticated approaches such as Fortune's algorithm\cite{de2000computational}
to compute the bounding box around tuple $t$ using the tuples from historic queries.

\begin{algorithm}[!htb]
\caption{{\bf Leverage-History}}
\begin{algorithmic}[1]
\label{alg:leverage_history}
\STATE {\bf Input:} $t$ and $H$ (set of tuples obtained from historic queries)
\STATE {\bf Output:} Bounding box $V'(t)$
\STATE $V'(t) = V_0$
\STATE {\bf for} each tuple $h \in H$
    \STATE \hindent Update $V'(t)$ with perpendicular bisector between $h$ and $t$ 
\STATE {\bf return} $V'(t)$ with the tightest bounding box around $t$
\end{algorithmic}
\end{algorithm}

{\bf Example 4:}
As part of computing $V(t_4)$ (see Example~1), we have the locations of $t_2, \ldots, t_5$.
Using this information, we can compute the initial bounding box for $t_2$ (shown in red around $t_2$ in Figure~\ref{fig:leverage_history})
{\em offline} - i.e. without issuing any queries.

\subsubsection{Variance reduction with larger $k$}\label{subsec:lrkgt1}

When the system has $k > 1$, we can of course still choose to use the top-1 Voronoi cell as if only the top result is returned. Or we can choose from any of the top-$h$ Voronoi cells as long as $h \leq k$. While intuitively it might appear that using all $k$ returned tuples is definitely better than using just the top-1, the theoretical analysis suggests otherwise - indeed, whether top-1 or top-$h$ Voronoi cell is better depends on the exact aggregate being estimated - specifically, whether the distribution of the attribute being aggregated is better ``aligned'' with the size distribution of top-1 or top-$h$ Voronoi cells. To see why, simply consider an extreme-case scenario where the aggregate being estimated is AVG(Salary), and the salary of each user (tuple) is exactly proportional to the size of its top-1 Voronoi cell. In this case, nothing beats using the top-1 Voronoi cells as doing so produces zero variance and thus zero estimation error.

Having said that, however, many aggregates can indeed be better estimated using top-$h$ Voronoi cells, because the sizes of these top-$h$ cells are more uniform than those of the top-1 cells, which can vary extremely widely (see Figure~\ref{fig:starbucksVoronoiDiag} in the experiments section for an example), while many real-world aggregates are also more uniformly distributed than the top-1 cell volume (again, see experiments for justification). But simply increasing $h$ also introduces an undesired consequence: recall from \S\ref{sec:pre} that the larger $h$ is, the more ``complex'' the top-$h$ Voronoi cell becomes - in other words, the more queries we have to spend in order to pin down the exact volume of the Voronoi cell.

Thus, the key is to make a proper tradeoff between the benefit received (i.e., smaller variance per sample) and the cost incurred (i.e., larger query cost per sample). Our main idea is a combination of two methods: leveraging history in \S\ref{sec:lhi} and upper/lower bound approximation in \S\ref{sec:lub}. Specifically, for each of the $k$ returned tuples, we perform the following process:

Consider $t_i$ returned as the No.~$i$ result. We need to decide which version of the Voronoi cell definition to use for $t_i$. The answer can be anywhere from 1 to $k$. To make the determination, for all $h \in [2, k]$, we compute $\lambda_h(t_i)$, the upper bound on the volume of the top-$k$ Voronoi cell of $t_i$, as computed from all historically retrieved tuples. Then, we choose the largest $h$ which satisfies $\lambda_h(t_i) \leq \lambda_0$, where $\lambda_0$ is a pre-determined threshold (the intuitive meaning of which shall be elaborated next). Let the chosen value be $h(t_i)$. If none of $h \in [2, k]$ satisfies the threshold, we take $h(t_i) = 1$. Then, if $h(t_i) \leq i$, we compute the top-$h$ Voronoi cell for $t_i$. The final estimation from the $k$ returned results becomes:
\begin{align}
\sum_{t_i: h(t_i) \leq i \leq k} \frac{Q(t_i)}{|V_h(t_i)|} \label{equ:es2}
\end{align}
for any SUM or COUNT query $Q$, where $|V_h(t_i)|$ is the volume for the top-$h$ Voronoi cell of $t_i$.

We now explain the intuition behind the above approach, specifically the threshold $\lambda_0$. First, note that if the top-$h$ (say top-1) Voronoi cell of $t_i$ is already large, then there is no need to further increase $h$. The reason can be observed from the above-described justification of variance reduction - note that a large top-1 Voronoi cell translates to a large selection probability $p$ - i.e., a small $1/p$ which adds little to the overall variance. Further increasing $h$ not only contributes little to variance reduction, but might actually increase the variance if $1/p$ is already below the average value.

Second, admittedly, $\lambda_h(t_i)$ is only an upper-bound estimate - i.e., even though we showed above that an already large top-$h$ Voronoi cell does not need to have $h$ further increased, there remains the possibility that $\lambda_h(t_i)$ is large because of an overly loose bound (from history), rather than the real volume of the Voronoi cell. Nonetheless, note that this is still a negative signal for using such a large $h$ - as it means that we have not thoroughly explored the neighborhood of $t_i$. In other words, we may need to issue many queries in order to reduce our estimation (or computation) of $|V_h(t_i)|$ from $\lambda_h(t_i)$ to the correct value. As such, we may still want to avoid using such a large $h$ in order to reduce the query cost.

While the above explanation is heuristic in nature 
it is important to note that, regardless of how we set $h(t_i)$, the estimation we produce for SUM and COUNT aggregates in (\ref{equ:es2}) is always unbiased.

\begin{algorithm}[!htb]
\caption{{\bf Variance-Reduction}}
\begin{algorithmic}[1]
\label{alg:variance_reduction}
\STATE {\bf Input:} $H$;  {\bf Output:} Aggregate estimate from all top-$k$ tuples 
\STATE $q$ = location chosen uniformly at random 
\STATE {\bf for} each tuple $t_i$ returned from query($q$)
    \STATE \hindent $h(t_i)$ = $\max \{h | h \in [2, k], \lambda_h(t_i) \leq \lambda_0 \}$
    \STATE \hindent $h(t_i) = 1$ if no $h$ satisfied the condition $\lambda_h(t_i) \leq \lambda_0$
    \STATE \hindent Generate estimate for $t_i$ using Equation~\ref{equ:es2}
\end{algorithmic}
\end{algorithm}

\subsubsection{Upper/lower bounds on Voronoi-cell} \label{sec:lub}

Note that in the entire process of Voronoi-cell computation (barring the very first step of the faster initialization idea discussed in \S\ref{sec:fip}), we maintain a tentative polygon that covers the entire real Voronoi cell - i.e., an upper bound on its volume. What often arises in practice, especially when computing top-$k$ Voronoi cells (which tend to have many edges), is that even though the bounding polygon is very close to the real Voronoi cell in volume, it has far fewer edges - meaning we still need to issue many more queries to pin down the exact Voronoi cell.

The key idea we develop here is to avoid such query costs {\em without} sacrificing the accuracy of our aggregate estimations.  Specifically, consider a simple Monte Carlo approach which chooses uniformly at random a point from the current bounding polygon, and then issues a query from that point. If the query returns $t$ - i.e., it is in the Voronoi cell of $t$, we stop. Otherwise, we repeat this process. Interestingly, the number of trials it takes to reach a point that returns $t$, say $r$, is an unbiased estimation of $|V^\prime(t)|/|V(t)|$, where $|V^\prime(t)|$ and $|V(t)|$ are the volumes of the bounding polygon and the real Voronoi cell of $t$, respectively. 
\begin{align*}
\mathrm{Exp}(r) &= \sum^{\infty}_{i=1} \left[i \cdot \left(1 - \frac{|V(t)|}{|V^\prime(t)|}\right)^{i-1} \cdot \frac{|V(t)|}{|V^\prime(t)|}\right] = 
\frac{|V^\prime(t)|}{|V(t)|}.
\end{align*}

In other words, we can maintain the unbiasedness of our estimation without issuing the many more queries required to pin down the exact Voronoi cell. Instead, when $V^\prime(t)$ is close enough to $V(t)$, we can simply use call upon above-described method which, in most likelihood, requires just one more query to produce an unbiased SUM or COUNT estimation. For example, we can simply multiply the number of trials $r$ by $|V_0|/|V^\prime(t)|$, where $|V_0|$ is the volume of the entire region under consideration, to produce an unbiased estimation for COUNT(*). Other SUM and COUNT aggregates can be estimated without bias in analogy.

Before concluding this idea, there is one more optimization we can use here: a lower bound on the top-$k$ Voronoi cell of $t$. In the following, we first discuss how to use such a lower bound to further reduce query cost, and then describe the idea for computing such a lower bound. Note that once we have knowledge of a region $R$ that is covered entirely by the real (top-$k$) Voronoi cell, if in the above process, we randomly choose a point $q$ (from $V^\prime(t)$) which happens to belong in $R$, then we no longer need to actually query $q$ - instead, we immediately know that $q$ must belong to $V(t)$ and can produce an unbiased estimation accordingly. This is the cost saving produced by knowledge of a lower bound $R$.

To understand how we construct this lower bound region, a key understanding is that, at anytime during the execution of our algorithm, we have tested certain vertices of $V^\prime(t)$ which are already confirmed to be part of $V(t)$. Consider such a vertex $v$. Let $C(v, t)$ be a circle with $v$ being the center and the distance between $t$ and $v$ being the radius. Note that we are guaranteed to have {\em observed all tuples} within $C(v, t)$. This essentially leads to a lower-bound estimation of $V(t)$. Specifically, a point $q$ is in this lower-bound region if and only if $C(q, t)$, i.e., a circle centered on $q$ with radius being the distance between $q$ and $t$, is entirely covered by the union of $C(v, t)$ for all vertices $v$ of $V^\prime(t)$ that have been confirmed to be within $V(t)$. As such, for any $q$ in this region, we can save the query on it in the above process.

{\bf Example 5:}
The upper bound $V'(t_4)$ of $V(t_4)$ after Step 3 in the Example~2 (i.e. run-through of Algorithm LR-LBS-AGG-Baseline)
is shown in Figure~\ref{fig:upper_lower_bounds} as a quadrilateral with red edges.
The three lower vertices of $V'(t_4)$ are guaranteed to be in $V(t_4)$ using the criteria described above 
and hence the polygon induced by them provides a lower bound estimate for $V(t_4)$.

\subsection{Algorithm LR-LBS-AGG}

By combining the baseline idea for precisely computing the Voronoi cells with the 4 techniques for error reduction,
we can design an efficient algorithm LR-LBS-AGG for aggregate estimation over LR-LBS.
Algorithm~\ref{alg:lrlbsagg} shows the pseudocode for LR-LBS-AGG.

\begin{algorithm}[!htb]
\caption{{\bf LR-LBS-AGG}}
\begin{algorithmic}[1]
\label{alg:lrlbsagg}
\STATE {\bf while} query budget is not exhausted
    \STATE \hindent $q$ = location chosen uniformly at random
    \STATE \hindent {\bf for} each tuple $t_i$ in query($q$)
        \STATE \hindent[2] Compute optimal $h$ for $t_i$
        \STATE \hindent[2] Construct initial $V_h(t_i)$ using Algorithms~\ref{alg:fast_init} and \ref{alg:leverage_history}
        \STATE \hindent[2] $D'$= vertices of $V_h(t_i)$ 
        \STATE \hindent[2] {\bf repeat} till $D'$ is not updated or Voronoi bound is tight
            \STATE \hindent[3] {\bf for} each vertex $v$ of $V_h(t_i)$: $D'=D'\cup$ query$(v)$
            \STATE \hindent[3] Update $V_h(t_i)$ and $V'_h(t_i)$ from $D'$ 
\STATE Produce aggregate estimation using samples
\end{algorithmic}
\end{algorithm}

\section{LNR-LBS-AGG}\label{sec:lnrlbsagg}

\subsection{Voronoi Cell Computation: Key Idea}\label{subsec:lnrk}

We now consider the case where only a ranked order of points are returned - but not their locations. We shall start with the case of $k = 1$, and then extend to the general case of $k > 1$. 

We start by defining a primitive operation of ``binary search'' as follows.  Consider the objective of finding the Voronoi cell of a tuple $t$ in the database. Given any location $c_1$ and $c_2$ (not necessarily occupied by any tuple), where $c_1$ returns $t$, consider the half-line from $c_1$ passing through $c_2$. Since a Voronoi cell is convex and $c_1$ resides within the Voronoi cell, this half-line has one and only one intersection with the Voronoi cell - which is associated with one or two edges of the Voronoi cell. We define the primitive operation of {\em binary search} for given $c_1, c_2$ to be the binary search process of finding one Voronoi edge associated with the intersection. Please refer to Appendix~\ref{sec:appendixBSearch} for the detailed design of this process.

Naturally, such a binary search process is associated with an error bound on the precision of the derived edge. For example, we can set an upper bound $\epsilon$ on the maximum distance between any point on the real Voronoi edge (i.e., a line segment) and its closest point on the derived edge, which we refer to as the {\em maximum edge error}, and use $\epsilon$ as the objective of the binary search operation. One can see that the number of queries required for this binary search is proportional to $\log(1/\epsilon)$. 
See Appendix~\ref{sec:appendixBSearch} for exact query cost.

Given this definition, we now show that one can discover the Voronoi cell of $t$ (up to whatever precision level afforded to us by the binary search operation) with a query complexity of $O(m \log(1/\epsilon))$, where $m$ is the number of edges for the Voronoi cell. Here is the corresponding process:

We start with one query at point $q$ which returns $t$. Then, we construct $4$ points that bound $q$ (say $q_1: \langle x(q)-1,y(q)\rangle$, $q_2: \langle x(q)+1,y(q)\rangle$, $q_3: \langle x(q),y(q)-1\rangle$, $q_4: \langle x(q),y(q)+1\rangle$, where $x(\cdot)$ and $y(\cdot)$ are the two dimensions, e.g., longitude and latitude, of a location, respectively) and call upon the binary search operation to find the corresponding Voronoi edges intercepting the half lines from $q$ to $q_1, \ldots, q_4$, respectively. One can see that, no matter what the discovered edges might be, they must form a closed polygon\footnote{In the extreme-case, some edges of this polygon might be part of the bounding box.} which we can use to initiate the testing process described in \S\ref{sec:lrk}. If all vertices pass the test, then we have already obtained the Voronoi cell of $t$. Otherwise, for each vertex (say $v$) that fails the test, we perform the binary search operation on the location of $v$ to discover another Voronoi edge. We repeatedly do so until all vertices pass the test - at which time we have obtained the real Voronoi cell - subject to whatever error bound specified for the binary search process (as described above).

To compute the query cost of this process, a key observation is that each call of the binary search process after the initial step (i.e., a call caused by a vertex failing the test) increases the number of discovered (real) edges for the Voronoi cell by 1. Thus, the number of times we have to call the binary search process is $O(m)$, leading to the overall query-cost complexity of $O(m \log(1/\epsilon))$. For the estimation error, we have the following theorem.
\begin{theorem}
The estimation bias for COUNT(*) is at most
\begin{align}
|E(\tilde{\theta} - \theta)| \leq \sum_{t \in D} \frac{\epsilon^2 - 2 \cdot d(t) \cdot \epsilon}{(d(t) - \epsilon)^2},
\end{align}
where $d(t)$ is the nearest distance between $t$ and another tuple in $D$, and $\epsilon$ is the aforementioned maximum edge error.
\end{theorem}

Estimation bias for other aggregates can be derived accordingly (given the distribution of the attribute being aggregated). One can make two observations from the theorem: First, the smaller maximum edge error $\epsilon$ is or the large inter-tuple distance $d(t)$ is, the smaller the bias will be. Second, we can make the bias arbitrarily small by shrinking $\epsilon$ - which leads to a log-scale increase of the query cost.

Algorithm~\ref{alg:lnrlbsagg} shows the pseudocode for LNR-LBS-AGG that also utilizes some of the error reduction ideas from \S\ref{sec:rer}.

\begin{algorithm}[!htb]
\caption{{\bf LNR-LBS-AGG}}
\begin{algorithmic}[1]
\label{alg:lnrlbsagg}
\STATE {\bf while} query budget is not exhausted
    \STATE \hindent $q$ = location chosen uniformly at random; $t$=query($q$)
    \STATE \hindent Construct four points $q_1, \ldots , q_4$ bounding $t$ 
    \STATE \hindent $e_i$ = Binary-Search($q_i$) $\forall i \in [1, 4]$
    \STATE \hindent $V(t)$ = closed polygon from Voronoi edges $e_1, \ldots, e_4$
    \STATE \hindent $D'$= vertices of $V(t)$ 
    \STATE \hindent {\bf repeat} till $D'$ is not updated
        \STATE \hindent[2] {\bf for} each vertex $v$ of $V(t)$: $D'=D'\cup$ query$(v)$
        \STATE \hindent[2] Find Voronoi edges $\forall d \in D'$ and update $V(t)$
\STATE Produce aggregate estimation using samples
\end{algorithmic}
\end{algorithm}

{\bf Example 6:}
We consider the same dataset as Example~1, except that in LNR-LBS the locations of tuples are not returned.
Figure~\ref{fig:lnr_lbs_agg_illstration} shows a run-through of the algorithm by which one of the Voronoi edges of $V(t_4)$ is identified.
Initially, the bounding box contains the entire region, i.e. $V_0$. 
$\ell_1$ and $\ell_2$ are two lines starting from $t_4$ constructed as per Algorithm~\ref{alg:bsearch}.
$p_1$ and $p_2$ are mid points of small line segments on $\ell_1$ and $\ell_2$ 
such that points on either side of them return different tuples when queried.
The new estimated Voronoi edge is computed as the line segment connecting $p_1$ and $p_2$.
Please refer to Appendix-\ref{sec:appendixBSearch} for further details.


\begin{figure*}[ht]                                                                                 
\begin{minipage}[t]{0.4\linewidth}                                                                  
    \centering                                                                                      
    \includegraphics[height=35mm,width=50mm]{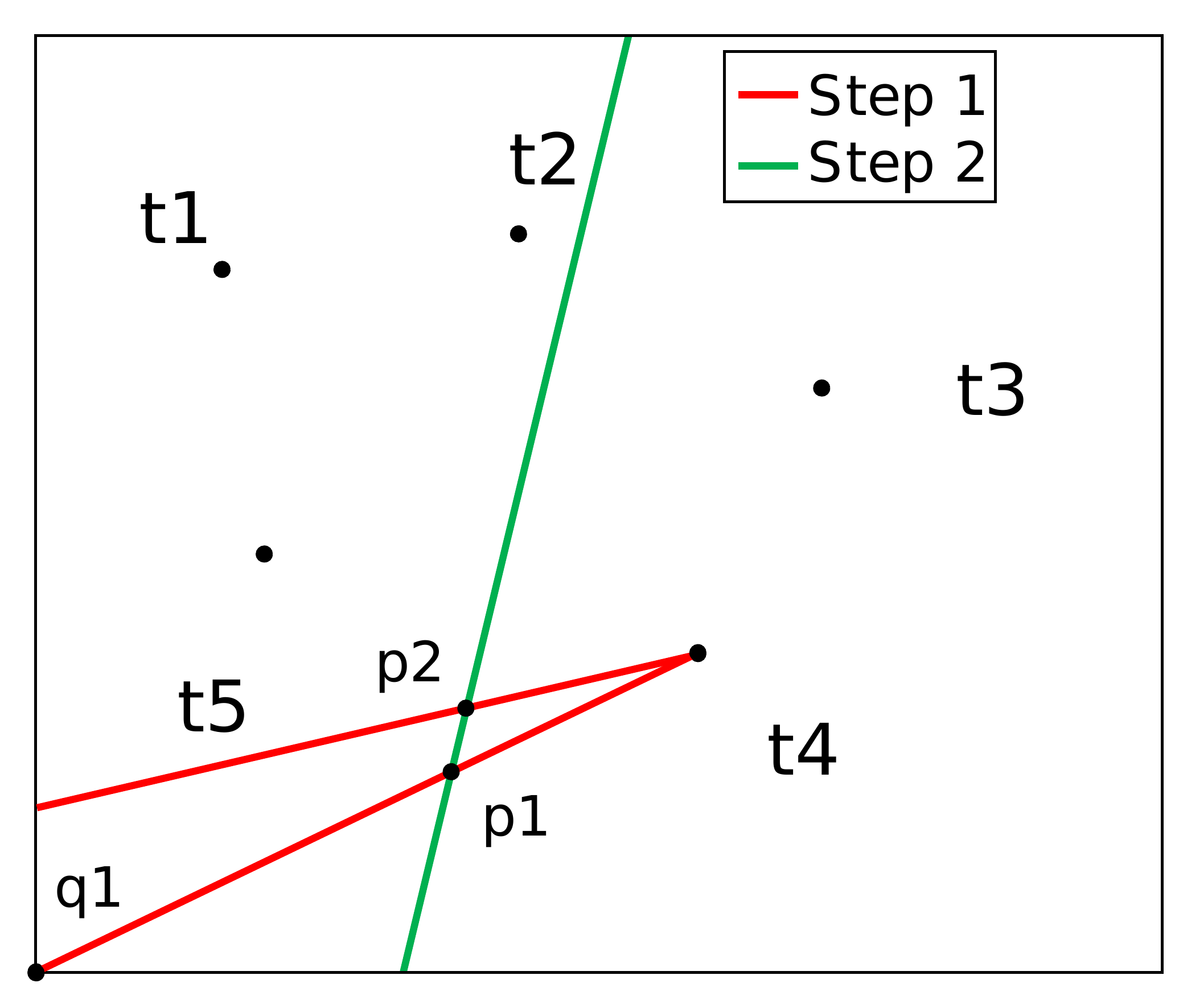}                                   
    \caption{Illustration of Algorithm LNR-LBS-AGG}                                                 
    \label{fig:lnr_lbs_agg_illstration}                                                             
\end{minipage}                                                                                      
\hspace{1mm}                                                                                        
\begin{minipage}[t]{0.55\linewidth}                                                                 
    \centering                                                                                      
    \subfigure{\includegraphics[height=30mm,width=40mm]{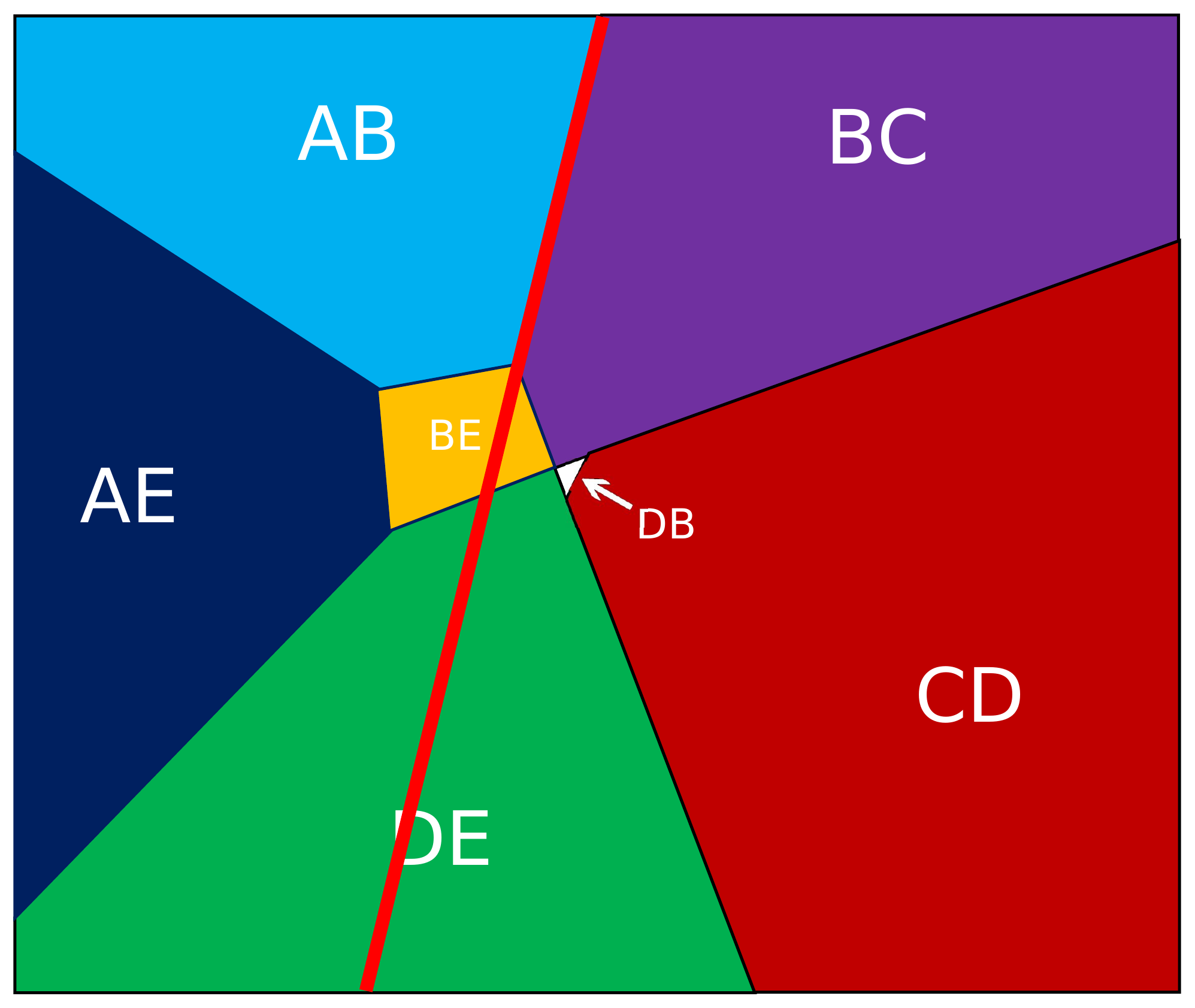}}\quad                  
    \subfigure{\includegraphics[height=30mm,width=40mm]{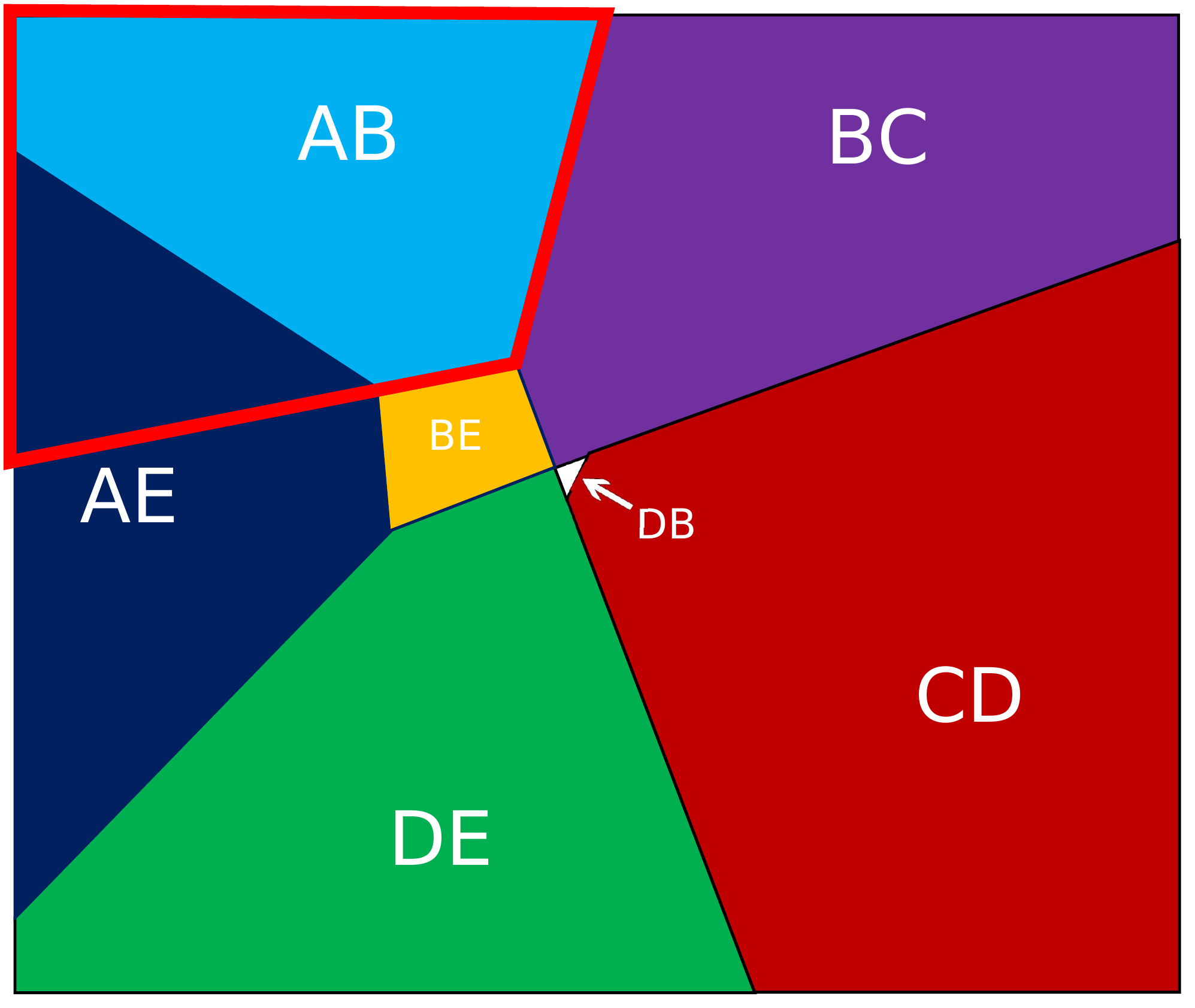}}                       
    \caption{Handling Concavity of top-$k$ Voronoi Diagrams}                                        
    \label{fig:lnrlbsaggConcave}                                                                    
\end{minipage}                                                                                      
\end{figure*}

\subsection{Extension to $k > 1$}\label{subsec:lnrkgt1}

A complication brought by the rank-only return semantics is the extension to cases with $k > 1$. Specifically, recall from \S\ref{sec:pre} that the (extended) top-$k$ Voronoi cell might be {\em concave} when $k > 1$. In the case LR-LBS case, this does not cause any problem because, at any moment, our derived top-$k$ Voronoi cell is computed from the exact tuple locations of all observed tuples and (therefore) completely covers the real top-$k$ Voronoi cell. For LNR-LBS case, however, this is no longer the case: Since we unveil the top-$k$ Voronoi cell edge after edge, if we happen to come across one of the ``concave edges'' early, then we may settle on a sub-region of the real top-$k$ Voronoi cell. Figure~\ref{fig:lnrlbsaggConcave} demonstrates an example for such a scenario.


Fortunately, there is an efficient fix to this situation. To understand the fix, a key observation is that any ``inward'' (i.e., concave) vertex of a top-$k$ Voronoi cell, say that of $t$, must be at a position with equal distance to three tuples, one of them being $t$ (Note: this might not hold for ``outward'' vertices). This property is proved in the following lemma.

\newtheorem{lemma}{Lemma}
\begin{lemma}
Any inward vertex of the top-$k$ Voronoi cell of $t$ must be of equal distance to $t$ and two other tuples in the database.
\end{lemma}
\begin{proof}
Consider a partition of the entire region into {\em base cells}, each of which returns a different combination of top-$k$ tuples. One can see that the top-$k$ Voronoi cell of $t$ must be the union of one or more adjacent base cells.  In addition, for general positioning (i.e., barring special positions such as bounding edges, etc.), any vertex of the top-$k$ Voronoi cell is formed by three edges (of some base cells in the partition). Now consider the three edges which form an inward vertex $v$, denoted by $e_1, e_2, e_3$.  Note that, given $v$ is inward, one of the three edges must be inside the top-$k$ Voronoi cell of $t$. Let this edge be $e_1$. One can see that both $e_2$ and $e_3$ separate the top-$k$ Voronoi cell from the outside - i.e., $\forall i \in \{1, 2\}$, we have locations on one side of $e_i$ returning $t$ in top-$k$ while locations on the other side do not. That is, each of $e_2$ and $e_3$ must be the perpendicular bisector of the line segment connecting $t$ and another tuple in the database. Let these two tuples be $t_2$ and $t_3$ for $e_2$ and $e_3$, respectively. In other words, $v$ must have equal distance to $t$, $t_2$ and $t_3$.
\end{proof}

Given this property, the extension to $k > 1$ becomes straightforward: Let $D^\prime$ be the set of all tuples we have observed which appear along with $t$ in the top-$k$ result of a query answer. Let $t \in D^\prime$. First, note that if the polygon we output is not the top-$k$ Voronoi cell of $t$, then it must be a sub-region of it missing at least one inward vertex. According to the above lemma, each inward vertex is formed by two perpendicular bisectors, each between $t$ and another tuple. A key observation here is that at least one of the {\em missed} inward vertices must be entirely formed by tuples in $D^\prime$.  The reason is simple: if no missed inward vertex satisfies this property, then we must have found the correct top-$k$ Voronoi cell of $t$ over $D^\prime$ - i.e., what we get so far must be a super-region of the correct top-$k$ Voronoi of $t$ over the entire database, contradicting our previous conclusion that it is a sub-region.

Now our task is reduced to finding such a missing inward vertex. Note that this is equivalent with finding the perpendicular bisector of $t$ and every other tuple in $D^\prime$ - as once these perpendicular bisectors are identified, the rest is simply getting their intersections which can be done offline. For each tuple in $D^\prime$, we either have already identified the perpendicular bisector through one of the previous calls to the binary search process - or we can initiate a new one as follows.

Specifically, to find the perpendicular bisector of $t$ and $t^\prime \in D^\prime$, note that $t^\prime$ being in $D^\prime$ means that (1) at least one of the vertices of the polygon we currently have must return $t^\prime$, and (2) at least one of the vertices of the polygon we currently have must not return $t^\prime$. In other words, there must exist an edge of our current polygon which has two vertices once returning $t^\prime$ and the other does not - i.e., this edge intercepts with the perpendicular bisector of $t$ and $t^\prime$. As such, we simply need to return the binary search process over this edge to find the perpendicular bisector, and then use it to update our polygon. We repeat this process iteratively until we have enumerated all perpendicular bisectors of $t$ and other nodes in $D^\prime$ - at which time we can conclude that there is no missing inward vertex. In other words, we have found the top-$k$ Voronoi cell of $t$. One can see that the query complexity of this process remains at $O(m \log(1/\epsilon))$, as every new binary search process called will return us a new edge for the top-$k$ Voronoi cell.

\subsection{Tuple Position Computation}\label{subsec:tuplePositionComputation}

Another important problem in the LNR-LBS case is the computation of a tuple's position, since such information is not returned in query answers as in the LR-LBS case. As discussed in the introduction, this problem can be of independent interest - it can also be called upon as a subroutine for aggregate query processing when the selection condition involves a tuple's location. For example, one might be interested in the number of WeChat users within 20 meters of major highways (i.e., those who are likely driving). To estimate this aggregate, we need to compute the location of a tuple (i.e., a WeChat user) in order to determine whether it satisfies the selection condition for the aggregate query.

\begin{figure}[ht]
\centering
\includegraphics[scale=0.4]{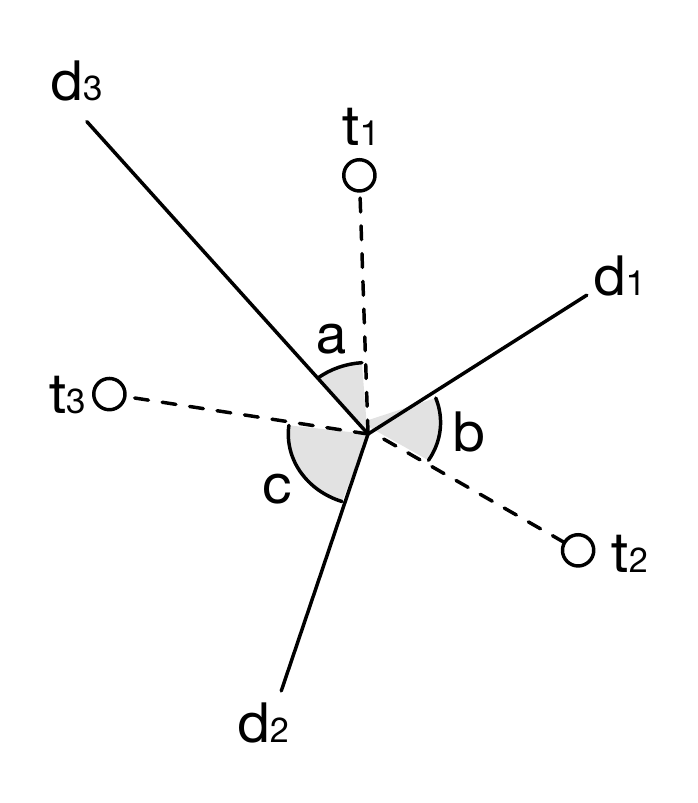}
\vspace{-5mm}
\caption{Demonstration of Tuple Position Computation}
\label{fig:r2l}
\end{figure}

Once we compute the Voronoi cell for a tuple $t$, the computation of $t$'s exact location takes only two additional calls to the binary search process. The key idea of this computation is demonstrated in Figure~\ref{fig:r2l}. The figure depicts one vertex of the top-1 Voronoi cell of $t_1$. Let the vertex (at the center of the figure) be the origin point $o$. The figure includes two edges of the Voronoi cell, $d_1$ and $d_3$, corresponding to the perpendicular bisector of $(t_1, t_2)$ and $(t_1, t_3)$, respectively. Note that since $o$ is of equal distance to $t_1$, $t_2$, and $t_3$, it must be attached to a third edge which is part of the Voronoi cell for $t_2$ and $t_3$ - this is depicted as $d_2$ in the figure.

In the following, we describe the computation of $t_1$'s location in three steps: First, we show that, with knowledge of $d_1$, $d_2$ and $d_3$, one can readily compute the line from $o$ to $t_1$ - i.e., the angle $a$ in the figure. Note that this indicates as long as one can do the same for another vertex of the Voronoi cell (say $o^\prime$), then the location of $t_1$ can be derived as the intersection of two lines: $(o, t_1)$ and $(o^\prime, t_1)$. Of course, in practice we only know $d_1$ and $d_3$ from the Voronoi cell computation, not $d_2$. Thus, we demonstrate in the second step that deriving $d_2$ from $d_1$ and $d_3$ takes only a single call to the binary search process.

First, to understand how angle $a$ can be derived from $d_1$, $d_2$, $d_3$, a key observation from Figure~\ref{fig:r2l} is that the lines from $o$ to any two tuples must form equal angle to the Voronoi edge between them - e.g., $(o, t_1)$ and $(o, t_3)$ must form equal angles to $d_3$. In other words, the angle between $(o, t_3)$ and $d_3$ is also $a$. Equipped with this observation, it becomes obvious that:
\begin{align}
(a + b) + (b + c) + (c + a) &= 2\pi\\
\Rightarrow a + b + c &= \pi
\end{align}
Since $b+c$ is exactly the angle between $d_1$ and $d_2$, we can easily compute $a$ as $\pi - (b+c)$. As such, we computed the line from $o$ to $t_1$ based on knowledge of only $d_1$, $d_2$ and $d_3$.

Now we explain how one can compute $d_2$ - the only one of the three edges not part of the Voronoi cell of $t_1$ - with a single call to the binary search process. Note from the fact that we have computed both $d_1$ and $d_3$ that we must have issues at least one query which returns $t_2$ as the top result, say $q_2$, and a query which returns $t_3$ on the top, say $q_3$. Obviously, $d_2$ intercepts the line segment between $q_2$ and $q_3$ exactly once. Thus, we simply need to call the binary search process over $(q_2, q_3)$ to derive $d_2$ and enable the computation of $t_1$'s exact location. One can see that, overall, the query complexity for computing both the Voronoi cell and the location of a tuple remains $O(m \log(1/\epsilon))$, where $m$ is the number of edges for its Voronoi cell.

\section{Discussions}\label{sec:ext}
\subsection{Aggregates with Selection Conditions}\label{subsec:selCondn}
In most of the previous discussions, we considered aggregates without selection conditions (i.e., every tuple in the bounding region is aggregated). There is indeed a straightforward extension to aggregates with selection conditions - specifically, there are two possible scenarios:

The first is when the selection condition can be ``passed through'' to LBS. For example, if our goal is to COUNT \textsc{``Starbucks''} within the bounding region, the selection condition \textsc{Name = `Starbucks'} can be passed through to LBS - i.e., we simply append to each query we issue the exact same selection condition as the aggregate, \textsc{Name = `Starbucks'}. One can see that no other change is required to the aggregate estimation process.

The other scenario is when the LBS does {\em not} support the selection condition. For example, if we want to COUNT all businesses with at least an average review score of four stars within the bounding region, then we cannot simply pass this selection condition to an LBS that does not support filtering by average review scores. In this case, we simply need to ``post-process'' the selection condition - e.g., for the above example, this means that after randomly choosing a query and obtain the returned tuple (as in \S\ref{sec:lrk}), we first determine if the tuple satisfies the filtering condition. If so, we continue with the original process and return the same estimation. Otherwise, we return 0 (i.e., the aggregate query applied over the returned tuple, again divided by the sampling probability) as the estimation. One can see that the result remains an unbiased estimation for the aggregate, now with selection conditions.

In the experiments, we shall demonstrate online tests over real-world LBS on aggregates with selection conditions in both categories - e.g., COUNT of \textsc{Starbucks} over Google Maps, which can be passed through, and COUNT(restaurants) that are open on Sundays, which cannot.

\subsection{Leveraging External Knowledge}

In previous discussions, we focused on how to process the results returned by a randomly chosen query (e.g., how to compute the top-$k$ Voronoi cell of a returned tuple). The way the initial query is chosen, however, remains a simple design of choosing a location uniformly at random from the bounding region. Admittedly, without any knowledge of the distribution of tuple locations, uniform distribution appears the natural choice. Nonetheless, in real-world applications, we often have certain {\em a priori} knowledge of the tuple distributions, which we can leverage to optimize the sampling distribution of queries.

For example, if our goal is to estimate an aggregate, say COUNT, of Point-Of-Interests (POIs, e.g., restaurants) in the US, a reasonable assumption is that the density of POIs in a region tends to be positively correlated with the region's population density. Thus, we have two choices: either to sample a location uniformly at random - which leads to POIs in rural areas to be returned with a much higher probability (because their Voronoi cells tend to be larger); or to sample a location with probability proportional to its population density - which hopefully leads to a more-or-less uniform selection probabilities over all POIs. Clearly, the second strategy is likely better for COUNT estimation, as a more uniform selection probability distribution directly leads to a smaller estimation variance (and therefore error). For example, in the extreme-case scenario where all POIs are selected with equal probability, our COUNT estimation will be precise with zero error. Thus, an optimization technique we adopt in this case is to design the initial sampling distribution of queries according to the population density information retrieved from external sources, e.g., US Census data \cite{uscensusdata}.

There are two important notes regarding this optimization: First, no matter if the external knowledge is accurate or not, the COUNT and SUM estimations we produce always remain unbiased. This is obvious from (\ref{equ:uce}) in \S\ref{sec:lrk}, which guarantees unbiasedness no matter what the sampling distribution $p(t)$ is. Second, the optimal sampling distribution depends on both the tuple distribution and the aggregate query itself. For example, if we want to estimate the SUM of review counts for all POIs, then the optimal sampling distribution is to sample each tuple with probability proportional to its review count (as this design produces zero estimation variance and error). Given the difficulty of predicting the aggregate (e.g., review COUNT in this case) ahead of time, leveraging external knowledge is better considered as heuristics (a very effective one nonetheless, as we shall demonstrate in experimental results) rather than a practice that guarantees the reduction of estimation errors.

\subsection{Special LBS Constraints}\label{subsec:lbsConstraints}

We now consider two special constraints that are enforced by the query interfaces of some real-world LBS. The first one is a {\em maximum radius} on the returned results - i.e., the distance between the query location $q$ and the returned tuples is bounded by a pre-determined threshold $d_{\max}$. If no tuple in the database falls within the circle centered at $q$ with radius $d_{\max}$, then the query returns empty. Google Maps and Weibo both enforce this constraint, with the threshold being 50 KM \cite{googleplacesapi} and 11 KM\footnote{\url{http://open.weibo.com/wiki/2/place/nearby/users}}, respectively.

Interestingly, no change is required for our algorithms (both LR-LBS-AGG and LNR-LBS-AGG) to handle this situation. One can see that, as long as a query result is non-empty, the nearest neighbor is always returned, enabling the usage of our algorithms. When a query returns empty, we simply return 0 as the COUNT or SUM estimation (for this sample query). The unbiasedness is unaffected - note from (\ref{equ:uce}) in \S\ref{sec:lrk} that unbiasedness is guaranteed no matter if the sampling probability $p(t)$ of all tuples sum up to 1 or not - as long as each tuple still has a positive probability to be returned. With this constraint, there is $\sum_t p(t) < 1$ with the remaining probability returning 0 - still leading to an unbiased SUM or COUNT estimation.

The second constraint we have observed from real-world LBS is a more complex ranking function that involves not only the distance between query location and a tuple but also other factors such as the static rank of certain attributes for the tuple. Google Places API is an example here, as it allows ranking by ``prominence'' which takes into account both distance and tuple popularity\footnote{Note that Google Places API also supports traditional distance-based ranking, enabling the direct usage of our algorithms.}.

For this constraint, the applicability of our results is no longer straightforward. The key challenge here is that the area returning a tuple may become segregated across many disjoint regions, making it extremely difficult to compute the sampling probability ($p(t)$ in (\ref{equ:uce}) in \S\ref{sec:lrk}) for a tuple. To understand why, consider an example where tuples are ranked according to the SUM of two scores, one is distance, awarding a higher score to a tuple closer by, but 0 to tuples more than 50 miles away. The other is a static score such as popularity. What might happen here is that the most popular tuple (in the bounding region, say US) is returned by queries on all places without a tuple within 50 miles (say the middle of a desert in Nevada). Clearly, it becomes extremely difficult to enumerate all the disjoint regions that return this tuple.

Fortunately, for LR-LBS in practice, it is still highly likely for our LR-LBS-AGG algorithm to successfully handle the constraint - because the algorithm works properly as long as the nearest neighbor is always included in the top-$k$ results. Since an LR-LBS returns tuple locations, we can always post-process the query answer to obtain the nearest neighbor according to distance, and then apply our algorithm. Given that $k \gg 1$ in real-world LBS, we anticipate a near-certain probability for the nearest neighbor to be included in the top-$k$ results, thus enabling LR-LBS-AGG.

\subsection{Extension to Higher Dimensions}

While LBS in practice is mostly confined to 2D, we would like to point out here (if only for theoretical interests) that our algorithm readily applies to $k$NN queries over higher-dimensional data where Euclidean distance is used as the ranking function. Specifically, note that for LR-LBS, Theorem~\ref{thm:lr} holds no matter what dimensionality the tuple locations have - as a higher-dimensional Voronoi cell computed from a subset of tuples still completely encompasses the real one. Similarly, all the optimizations discussed in \S\ref{sec:rer} readily apply as well. For LNR-LBS-AGG, the only change required is on the binary search process: instead of finding the perpendicular bisecting {\em line} between two tuples as in the 2D case, we now need to find the perpendicular bisecting $(d-1)$-dimensional plane in the $d$-dimensional case. Correspondingly, each vertex of the $d$-D Voronoi cell is now the interception of $\left(d \atop 2\right)$ such $(d-1)$-dimensional planes. In other words, we still only need two vertices of the Voronoi cell to derive a tuple's location in LNR-LBS - enabling the usage of LNR-LBS-AGG.

\section{Experimental Results}\label{sec:exp}

\subsection{Experimental Setup}\label{sec:expSetup}
\noindent{\bf Hardware and Platform:}
All our experiments were performed on a quad-core 2.5 GHz Intel i7 machine running Ubuntu 14.10 with 16 GB of RAM.
The algorithms were implemented in Python.

\noindent{\bf Offline Real-World Dataset:}
To verify the correctness of our results, we started by testing our algorithms locally over OpenStreetMap \cite{openstreetmap}, a real-world spatial database consisting of POIs (including restaurants, schools, colleges, banks, etc.) from public-domain data sources and user-created data.

We focused on the USA portion of OpenStreetMap. To enrich the SUM/COUNT/AVG aggregates for testing, we grew the attributes of POIs (specifically, restaurants and schools) by ``JOINing'' OpenStreetMap with two external data sources, Google Maps\cite{googleplacesapi} and US Census \cite{uscensusdata}. Specifically, we added for each (applicable) restaurant POI its review ratings from Google Maps; and each school POI its enrollment number from US Census. The US Census data is also used as the (optional) external knowledge source - i.e., to provide the population density data for the optimization technique discussed in \S\ref{sec:ext}.

Note that we have complete access to the enriched dataset and full control over its query interface. Thus, we implemented a $k$NN interface with ranking function being the Euclidean distance; returned attributes either containing all attributes including location (for testing LR-LBS) or without location (for LNR-LBS); and varying $k$ to observe the change of performance for our algorithms.

\vspace{1mm}
\noindent{\bf Online LBS Demonstrations:}
In order to showcase the efficacy of our algorithms in real-world applications, we also conducted experiments {\em online} over 3 very popular real-world LBS: Google Maps \cite{googleplacesapi}, WeChat\cite{wechat}, and Sina Weibo\cite{weibo}. Each of these services has at least hundreds of millions of users. Unlike the offline experiments, we do not have direct access to the ground-truth aggregates due to the lack of partnership with these LBS. Nonetheless, we did attempt to verify the accuracy of our aggregate estimations with information provided by external sources (e.g., news reports) - more details later in the section.

In online experiments for LR-LBS, we used Google Maps, specifically its Google Places API \cite{googleplacesapi}, which takes as input a query location (latitude/longitude pair) and (optionally) filtering conditions such as keywords (e.g., ``Starbucks'') or POI type (e.g., ``restaurant''), and returns at most $k = 60$  POIs nearby, ordered by distance from low to high, with location and other relevant information (e.g., review ratings) returned for each POI.

For LNR-LBS, we tested WeChat and Sina Weibo using their respective Android apps. Both directly fetch locations from the OS positioning service and search for nearby users, with WeChat returning at most $k = 50$ and Sina Weibo returning $k = 100$ nearest users. Unlike Google Maps, these two services do {\em not} return the exact locations of these nearby users - but only provide attributes such as name, gender, etc.

An implementation-related issue regarding WeChat is that, unlike its mobile apps, its API does not support nearest-neighbor search. Thus, we conducted our experiments by running its Android app (with support for nearest-neighbor search) on the official Android emulator, and used the debugging feature of location spoofing to issue queries from different locations. We then used the MonkeyRunner tool\footnote{\url{http://developer.android.com/tools/help/monkeyrunner_concepts.html}} for Android emulator to interact with the app - i.e., sending queries and receiving results. Specifically, to extract query answers from the Android emulator, we first took a screenshot of the query-answer screen, and then parsed the results through an OCR engine, tesseract-ocr\footnote{\url{https://code.google.com/p/tesseract-ocr/}}.

\vspace{1mm}
\noindent{\bf Algorithms Evaluated:}
We mainly evaluated three algorithms in our experiments: LR-LBS-AGG and LNR-LBS-AGG from \S\ref{sec:lrlbsagg} and \S\ref{sec:lnrlbsagg}, respectively, along with the only existing solution for LR-LBS (note there is no existing solution for LNR-LBS), which we refer to as LR-LBS-NNO \cite{DKA+11}.  LR-LBS-NNO has a number of tuneable parameters - we picked the parameter settings and optimizations from \cite{DKA+11} that provided the best performance. We also tested variants of our algorithms that lack certain variance-reduction techniques discussed in \S\ref{sec:lrlbsagg} and the weighted sampling in order to demonstrate the effectiveness of these techniques.


\vspace{1mm}
\noindent{\bf Performance Measures:} As discussed in \S\ref{sec:pre}, we measure efficiency through query cost, i.e., the number of queries issued to the LBS. Our estimation accuracy is measured experimentally by relative error. Each data point is obtained as the average of 25 runs.

\subsection{Experiments over Real-World Datasets}\label{subsec:expOffline}
\begin{figure}[t]
    \centering
    \includegraphics[scale=0.18]{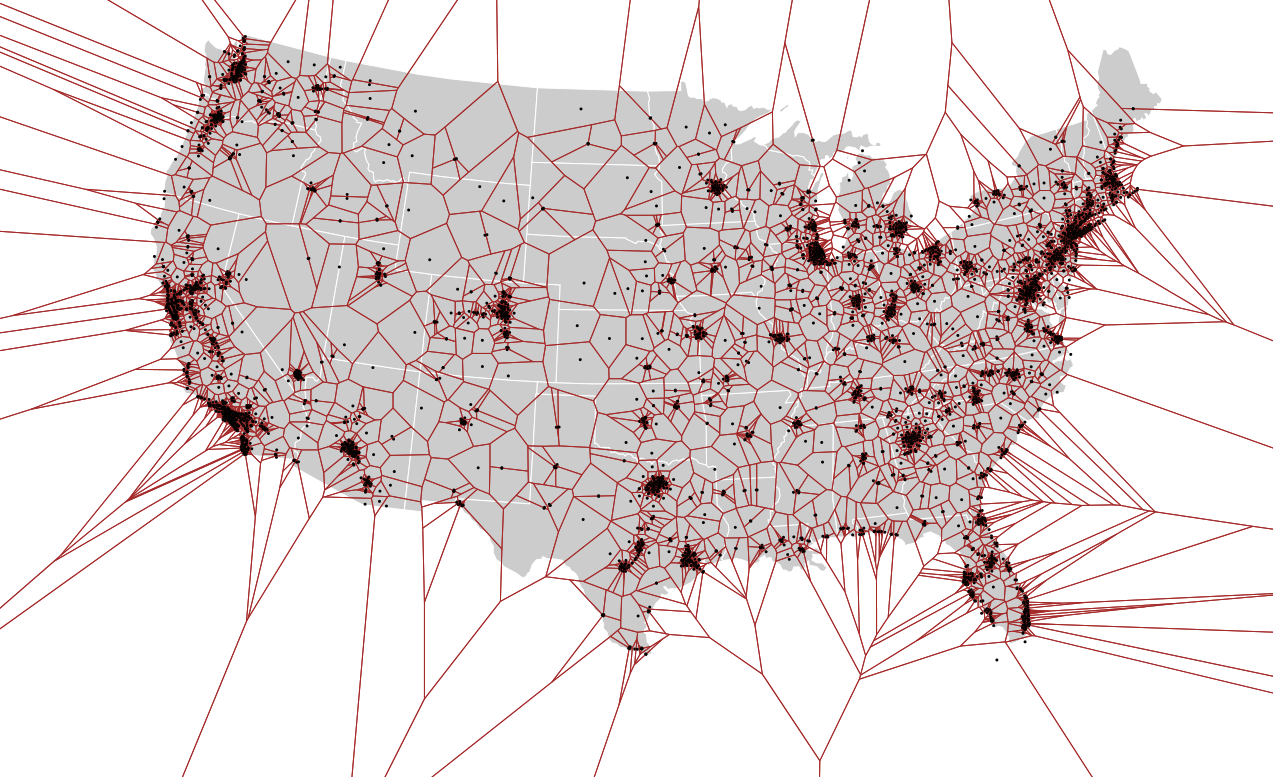}
    \caption{Voronoi Decomposition of Starbucks in US}
    \label{fig:starbucksVoronoiDiag}
\end{figure}

\noindent {\bf Unbiasedness of Estimators:}
Our first experiment seeks to show the unbiasedness of our estimators for LR-LBS-AGG and LNR-LBS-AGG even after incorporating the various error reduction strategies.
LR-LBS-NNO is known to be unbiased from \cite{DKA+11} after an expensive bias correction step.
Figure~\ref{fig:algo_conv_rate} shows a trace of the three algorithms when estimating COUNT of all restaurants in US
by plotting the current estimate periodically after fixed number of queries have been issued to LBS.
We can see that LR-LBS-NNO has a high variance and takes significantly longer to converge
while our estimators quickly converge to the ground truth much before LR-LBS-NNO.
This indicates that the error reduction techniques successfully reduce the variance of our estimators.

\noindent {\bf Query Cost versus Relative Error:} We start by testing the key tradeoff - i.e., query cost vs.~relative error - for all three algorithm over various aggregates. Specifically, Figures~\ref{fig:ws_relErrorVsQC_Count}, \ref{fig:ws_relErrorVsQC_Count_restaurants}, \ref{fig:ws_relErrorVsQC_sum} and \ref{fig:ws_relErrorVsQC_avg}
show the results for four queries, COUNT of schools in US, COUNT of restaurants in US, SUM of school enrollments in US, and AVG of restaurant ratings in Austin, Texas, respectively. One can see that not only our LR-LBS-AGG algorithm significantly outperform the previous LR-LBS-NNO \cite{DKA+11} in all cases, even our algorithm for the LNR-LBS case achieves much better performance than the previous algorithm (despite the lack of tuple locations in query results).

\noindent {\bf Query Cost versus LBS Size:}
Figure~\ref{fig:ws_NVsQC_Count} shows the impact of LBS database size (in terms of number of POIs or users)
on query cost to estimate the COUNT of schools in US for a fixed relative error of 0.1 .
We varied the database size by picking a subset of the database (such as 25\%, 50\%, etc) uniformly at random and estimating the aggregate over it.
As expected for a sampling-based approach, the increase in database size do not have any major impact and only results in a slight increase in overall query cost (due to the more complex topology of Voronoi cells).

\noindent {\bf Query Cost versus $k$:}
Figure~\ref{fig:ws_KVsQC_Count} shows how the value of $k$ (the number of tuples returned by $k$-NN interface) affects the query cost.
Again, we measure the query cost required to achieve a relative error of 0.1 on the aggregate COUNT of schools in US.
We compared an variant that leverages our variance reduction strategy that adaptively decides which subset of tuples (i.e. $h$ of top-$k$) to use
with fixed variants that uses all the top-$k$ tuples.
As expected, our adaptive strategy has a lower query cost and consistently achieves a saving of 10\% of query cost.

\noindent {\bf Efficacy of Error Reduction Strategies:}
We started by verifying the effectiveness of weighted sampling using external knowledge - Figure~\ref{fig:weighted_uniform_sampling_count} compares the performance of the two sampling strategies -
uniform and weighted - while estimating the COUNT of schools in US. One can see that the weighted sampling variants result in significant savings in query cost.

In our final set of experiments, we evaluated the efficacy of the various error reduction strategies we described in the paper.
We compared 5 different variants of our algorithm for LR-LBS
ranging from no error reduction strategies (LR-LBS-AGG-0) to sequentially adding them one by one in the order discussed in the paper culminating in LR-LBS-AGG that incorporates all of them. Figure~\ref{fig:ws_seq_opt} shows the results of this experiment.
As expected the first two strategies of faster initialization and leveraging history caused a significant reduction in query cost.
We observed that the results for LNR-LBS were very similar.

\begin{figure*}[ht]
\begin{minipage}[t]{0.23\linewidth}
\centering
    \includegraphics[width = 45mm, height = 32mm]{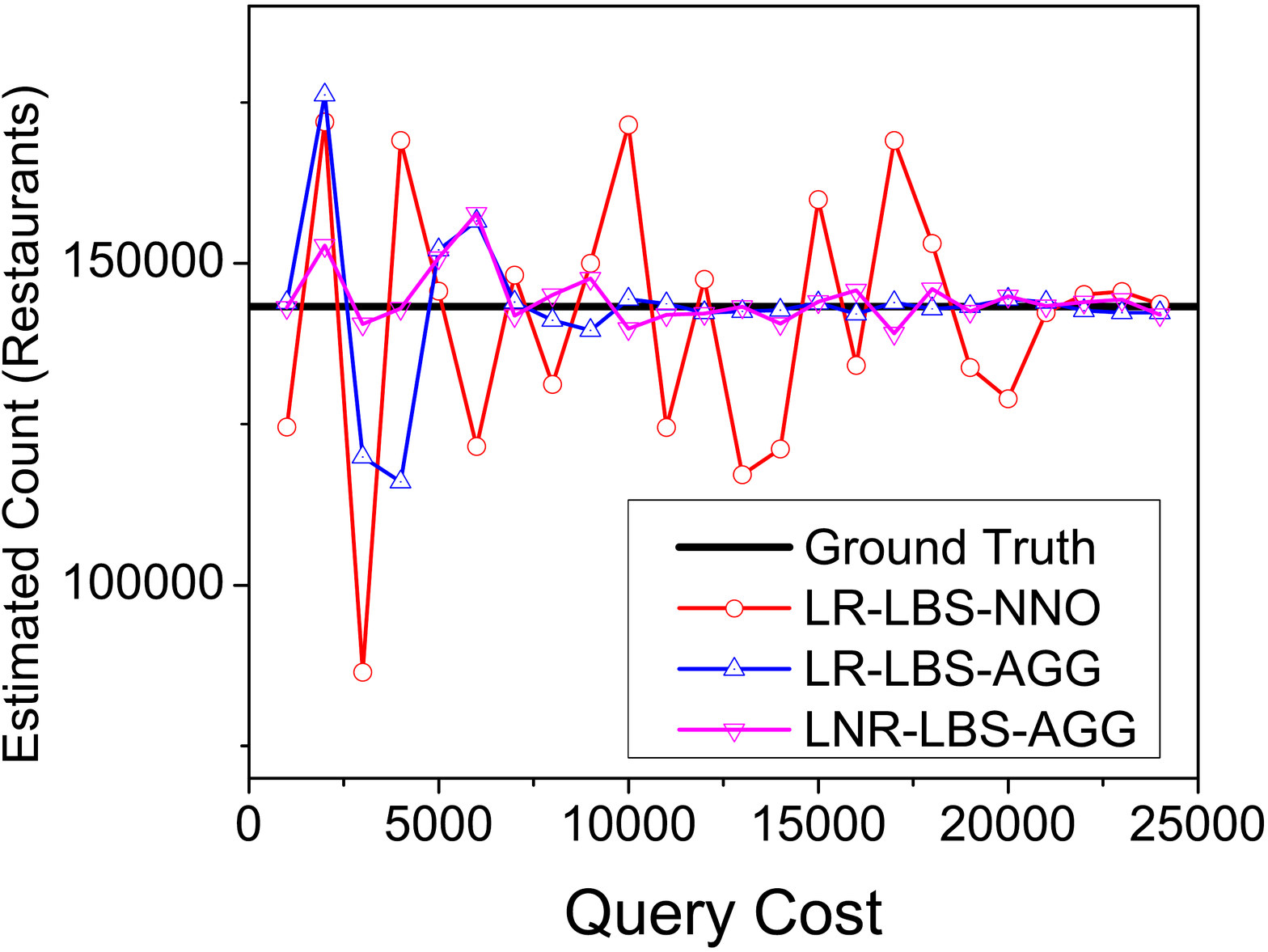}
    \vspace{-7mm}\caption{Unbiasedness of Estimators}
    \label{fig:algo_conv_rate}
\end{minipage}
\hspace{1mm}
\begin{minipage}[t]{0.23\linewidth}
\centering
    \includegraphics[width = 45mm, height = 32mm]{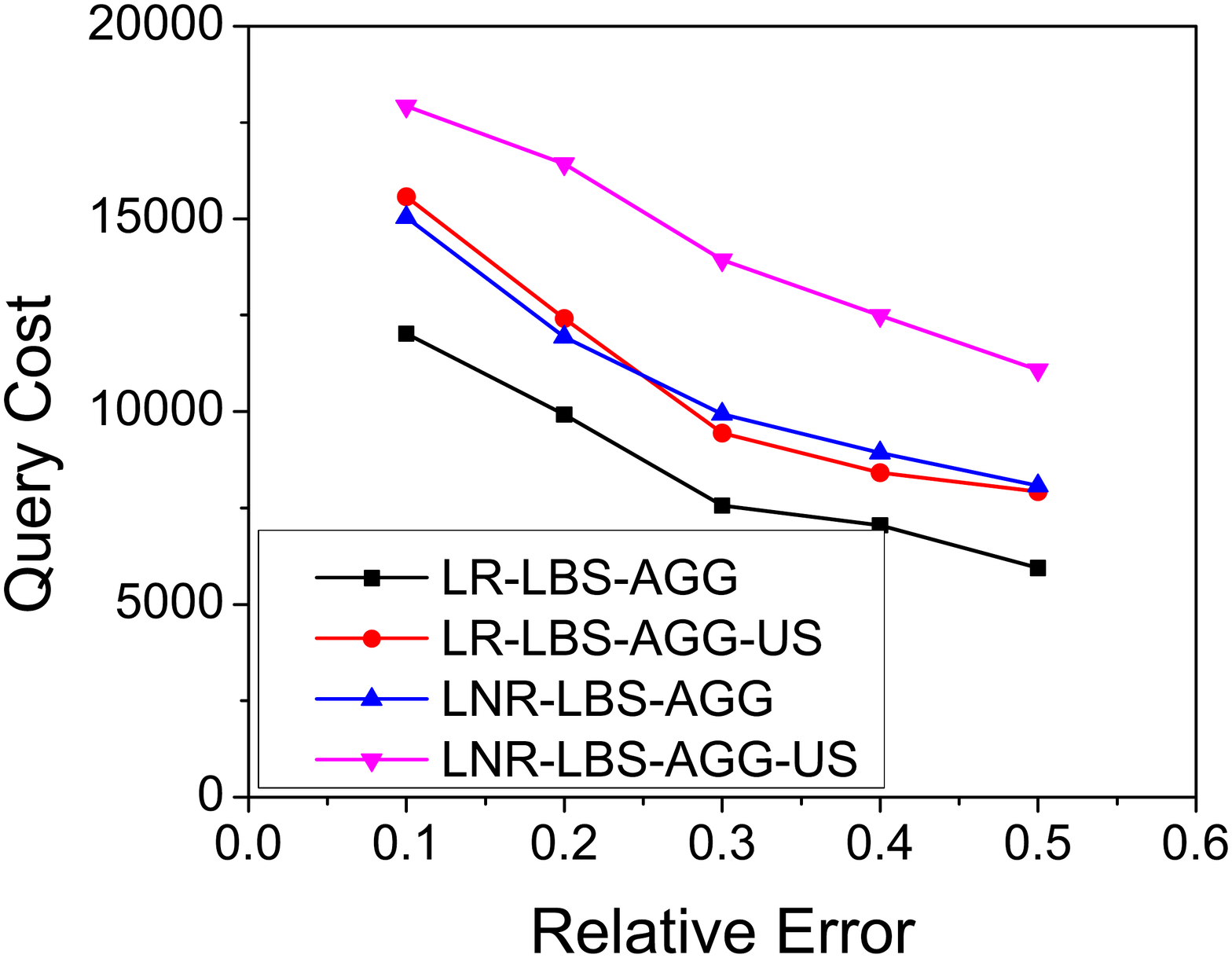}
    \vspace{-7mm}\caption{Impact of Sampling Strategy}
    \label{fig:weighted_uniform_sampling_count}
\end{minipage}
\hspace{1mm}
\begin{minipage}[t]{0.23\linewidth}
\centering
    \includegraphics[width = 45mm, height = 32mm]{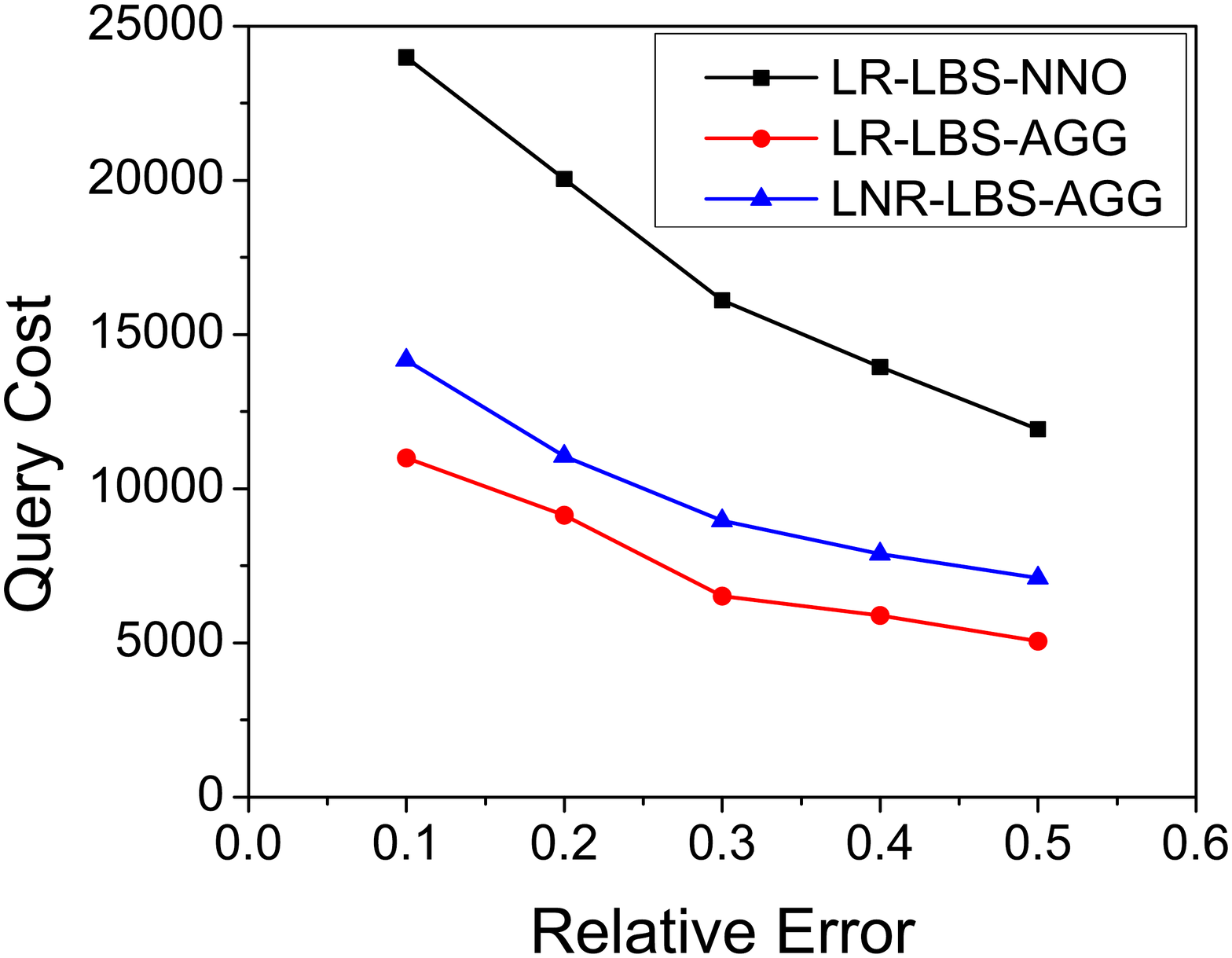}
    \vspace{-7mm}\caption{COUNT(schools)}
    \label{fig:ws_relErrorVsQC_Count}
\end{minipage}
\hspace{1mm}
\begin{minipage}[t]{0.23\linewidth}
\centering
    \includegraphics[width = 45mm, height = 32mm]{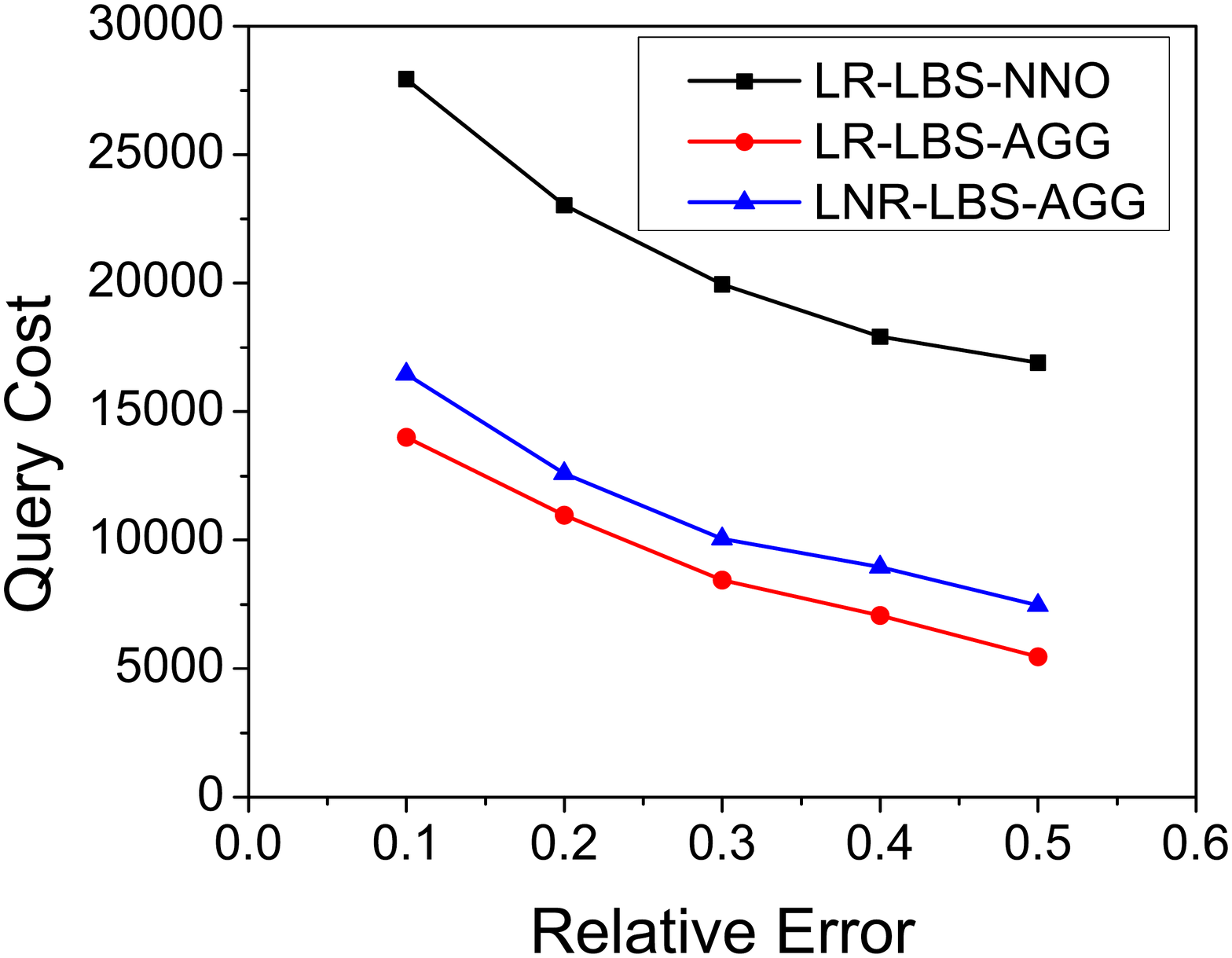}
    \vspace{-7mm}\caption{COUNT(restaurants)}
    \label{fig:ws_relErrorVsQC_Count_restaurants}
\end{minipage}
\hspace{-2mm}
\end{figure*}
\vspace{-4mm}
\begin{figure*}[ht]
\begin{minipage}[t]{0.23\linewidth}
\centering
    \includegraphics[width = 45mm, height = 32mm]{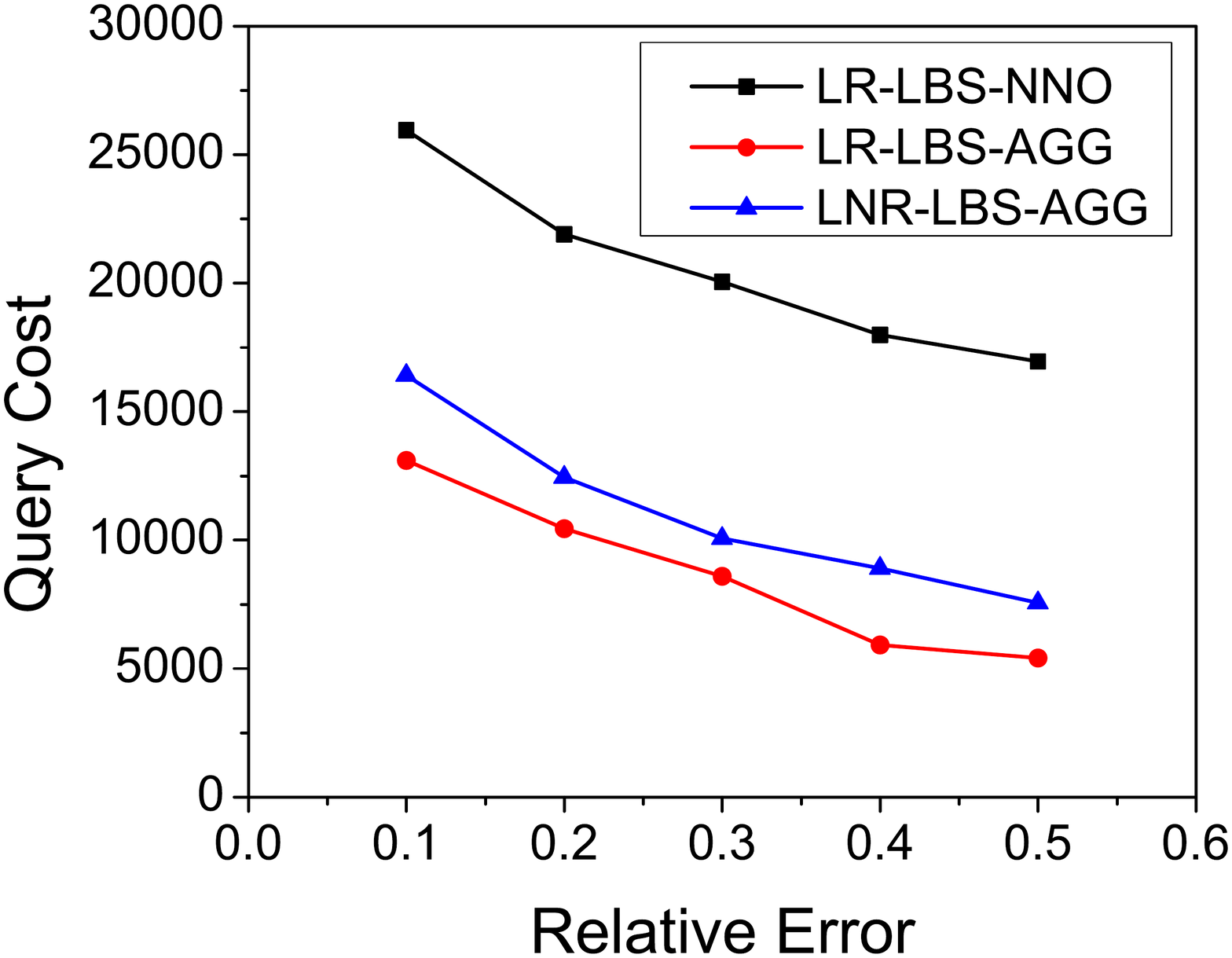}
    \vspace{-7mm}\caption{SUM(enrollment) in Schools}
    \label{fig:ws_relErrorVsQC_sum}
\end{minipage}
\hspace{1mm}
\begin{minipage}[t]{0.23\linewidth}
\centering
    \includegraphics[width = 45mm, height = 32mm]{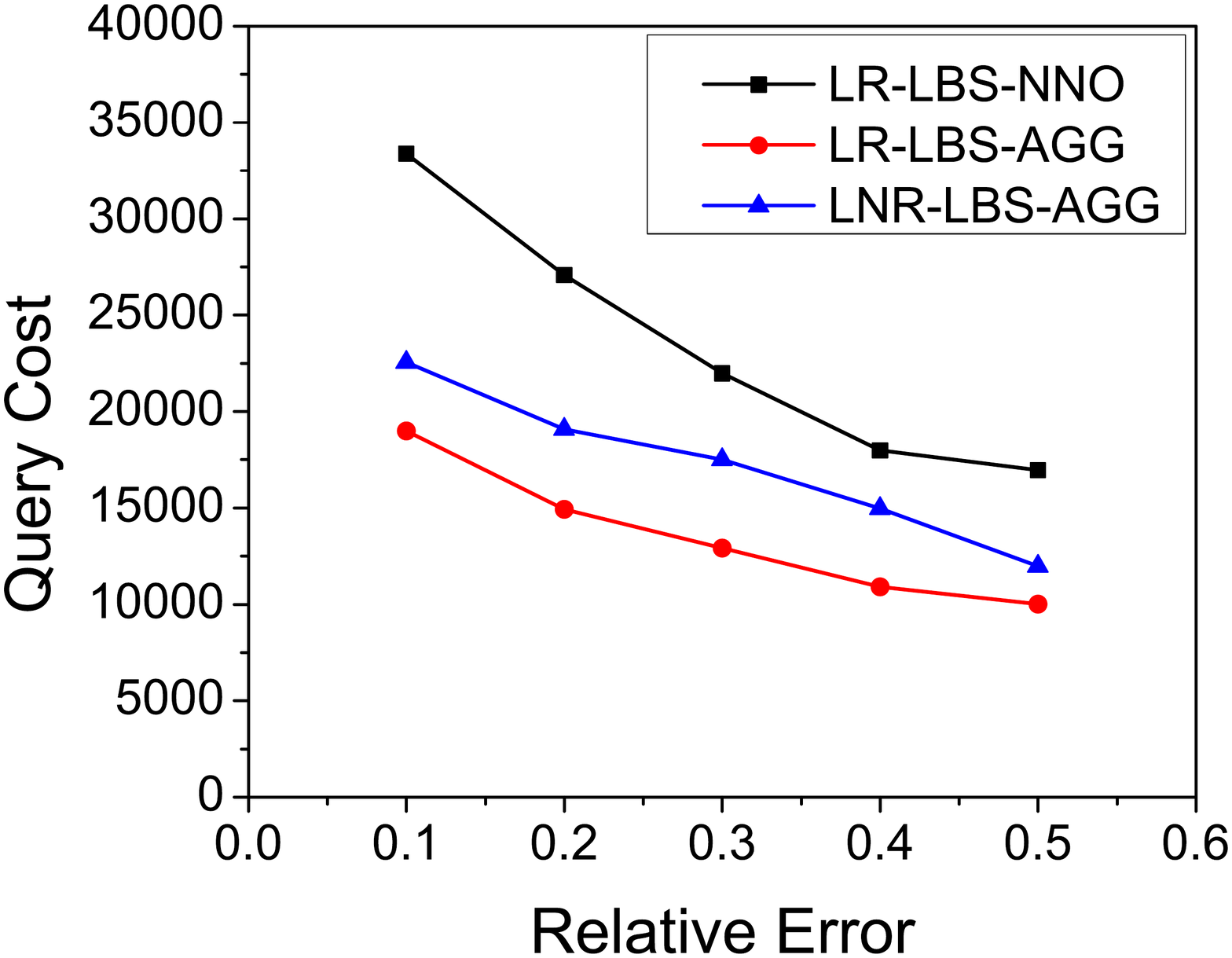}
    \vspace{-7mm}\caption{AVG(ratings) in Austin, TX Restaurants}
    \label{fig:ws_relErrorVsQC_avg}
\end{minipage}
\hspace{1mm}
\begin{minipage}[t]{0.23\linewidth}
\centering
    \includegraphics[width = 45mm, height = 32mm]{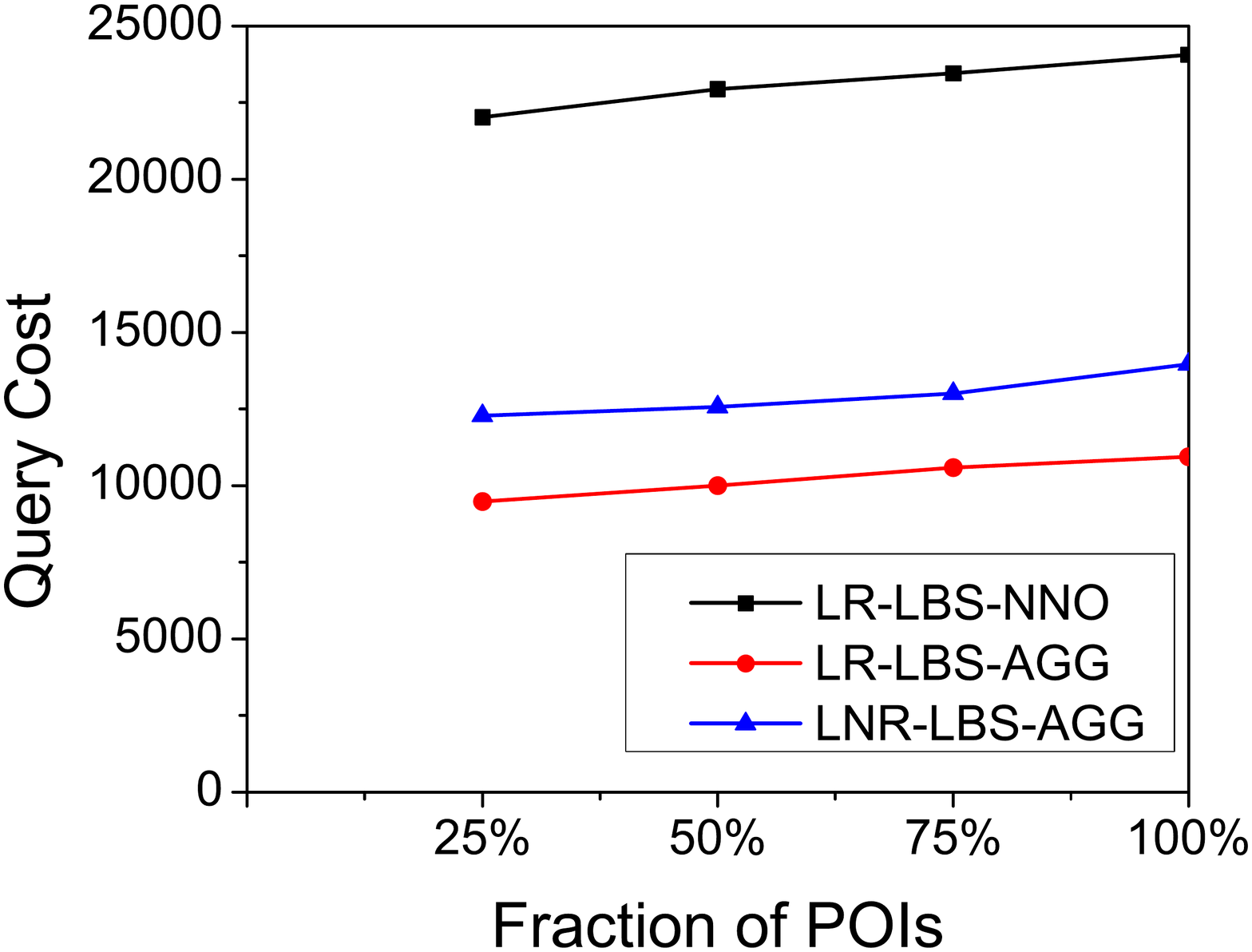}
    \vspace{-7mm}\caption{Varying Database Size}
    \label{fig:ws_NVsQC_Count}
\end{minipage}
\hspace{1mm}
\begin{minipage}[t]{0.23\linewidth}
\centering
    \includegraphics[width = 45mm, height = 32mm]{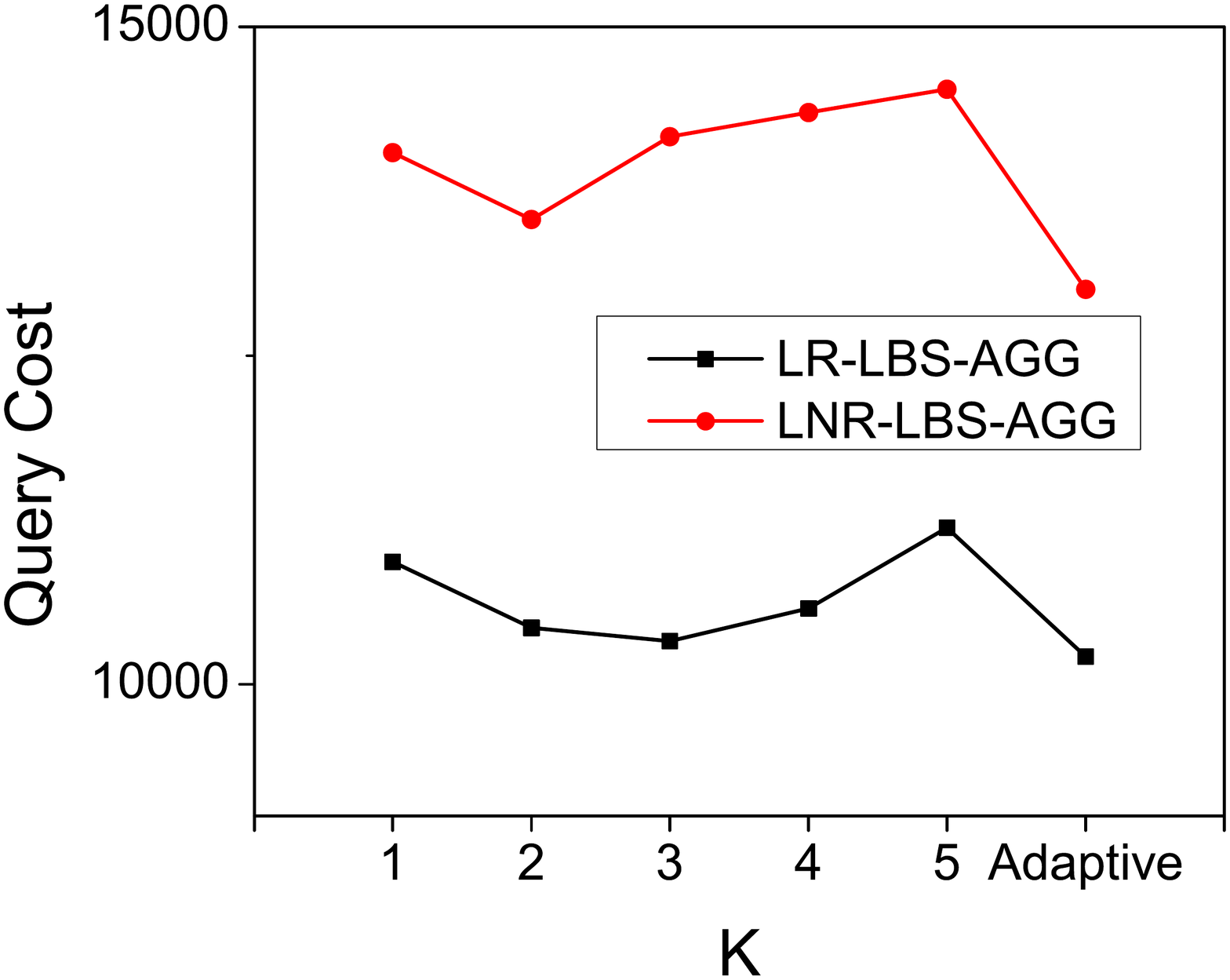}
    \vspace{-7mm}\caption{Varying $k$}
    \label{fig:ws_KVsQC_Count}
\end{minipage}
\hspace{-2mm}
\end{figure*}

\vspace{2mm}
\subsection{Online Demonstrations}\label{subsec:expOnline}

\noindent {\bf Google Places:}
Our first online demonstration of LR-LBS-AGG was on Google Places API and estimating two aggregates with different selection conditions. The first involves selection conditions that can be passed over to LBS (COUNT of Starbucks in US) while the second involves aggregates with selection condition that cannot be passed over (see \S\ref{sec:ext} for discussion) such as COUNT of restaurants in Austin, Texas that are open on Sundays.

Table~\ref{tbl:onlineExp} shows the results of the experiments. We also verified the accuracy of our estimates for first aggregate (COUNT of Starbucks) through the public release of Starbucks Corp\cite{starbucksRef}. One can see from the table that, with just 5000 queries, LR-LBS-AGG achieves very accurate estimations ($< 5\%$ relative error) for the count.

To provide an intuitive illustration of the execution of our algorithm, we also continued the estimation of COUNT(``Starbucks'') until enumerating all Starbucks in the US. Figure~\ref{fig:starbucksVoronoiDiag} demonstrates the Voronoi diagram constructed by our algorithm at the end. One can see the vastly different sizes of Voronoi cells - spanning hundreds of thousands km$^2$ in rural areas and smaller than 1km$^2$ in urban cities, justifying the effectiveness of weighted sampling.

\vspace{1mm}
\noindent {\bf WeChat and Sina Weibo:}
We estimated two aggregates, (1) total number of users and (2) gender ratio, over two LNR-LBS, WeChat and Sina Weibo, respectively. Table~\ref{tbl:onlineExp} shows the results of the experiments. One can observe from the table that our estimations quickly converge to a narrow range (+/- 5\%) after issuing a small number of queries (10000). While we do not have access to the ground truth this time, we do note an interesting observation from our results: the percentage of male users is much higher on WeChat than on Sina Weibo - an observation verified by various surveys in China \cite{wechatweiboref}.
We would like to note that the COUNT aggregate measures the number of users who have enabled the location feature of WeChat and Weibo respectively
and is different from the number of registered or active accounts. 


\begin{table}[t]
\small
\centering
\caption{Summary of Online Experiments}
\label{tbl:onlineExp}
\begin{tabular}{p{1cm} p{3cm} c c} 
  \hline
    {\bf LBS}             &   {\bf Aggregate}               &  {\bf Estimate}  &  {\bf Query Budget} \\ \hline
    Google Places   &  COUNT(Starbucks in US)    & 12023      &  5000        \\ \hline 
    Google Places   &  COUNT(restaurants in Austin TX and open on Sundays)   & 2856 &  5000   \\ \hline
    WeChat          & COUNT(WeChat users in China) & 338.4 M & 10000 \\ \hline
    WeChat & Gender Ratio of WeChat users in China & 67.1:32.9 & 10000 \\ \hline
    Weibo & COUNT(Weibo users in China) & 44.6 M & 10000 \\ \hline
    Weibo & Gender Ratio of Weibo users in China & 50.4:49.6 & 10000 \\ \hline
\end{tabular}
\end{table}

\vspace{1mm}
\noindent {\bf Localization Accuracy:}
As a final set of experiments, we also evaluated the effectiveness of our Tuple position computation approaches in tracking real world users.
Specifically, we sought to identify the precise location of {\em static} objects located across the region.
We conducted this experiment over Google Places in US and WeChat in China.
We treated Google Places as LNR-LBS by ignoring the location provided its API.
We sought to identify the location of 200 randomly chosen POIs after issuing at most 100 queries for each POI.
For WeChat, we positioned our user at 200 diverse locations within China (typically in Urban places) and sought to identify the location.
Since the precise location of the POI/user is known, we can compute the distance between actual and estimated positions.
Figure~\ref{fig:tracking_accuracy} shows the result of the experiments.
The results show that more than 80\% of the POIs were located within 20m of the exact location and every POI was located within a distance of 75m.
Due to the various location obfuscation strategies employed by WeChat, we achieved an accuracy of 50m or lower only 45\% of the time.
We still were able to locate user within 100m  almost all the time.
While our theoretical methods could precisely identify the location, the discrepancy in real-world occurs due to
various external factors such as obfuscation, coverage/localization limits etc.

\begin{figure}[h]
\begin{minipage}[h]{0.48\linewidth}
    \centering
    \includegraphics[width = 45mm, height = 32mm]{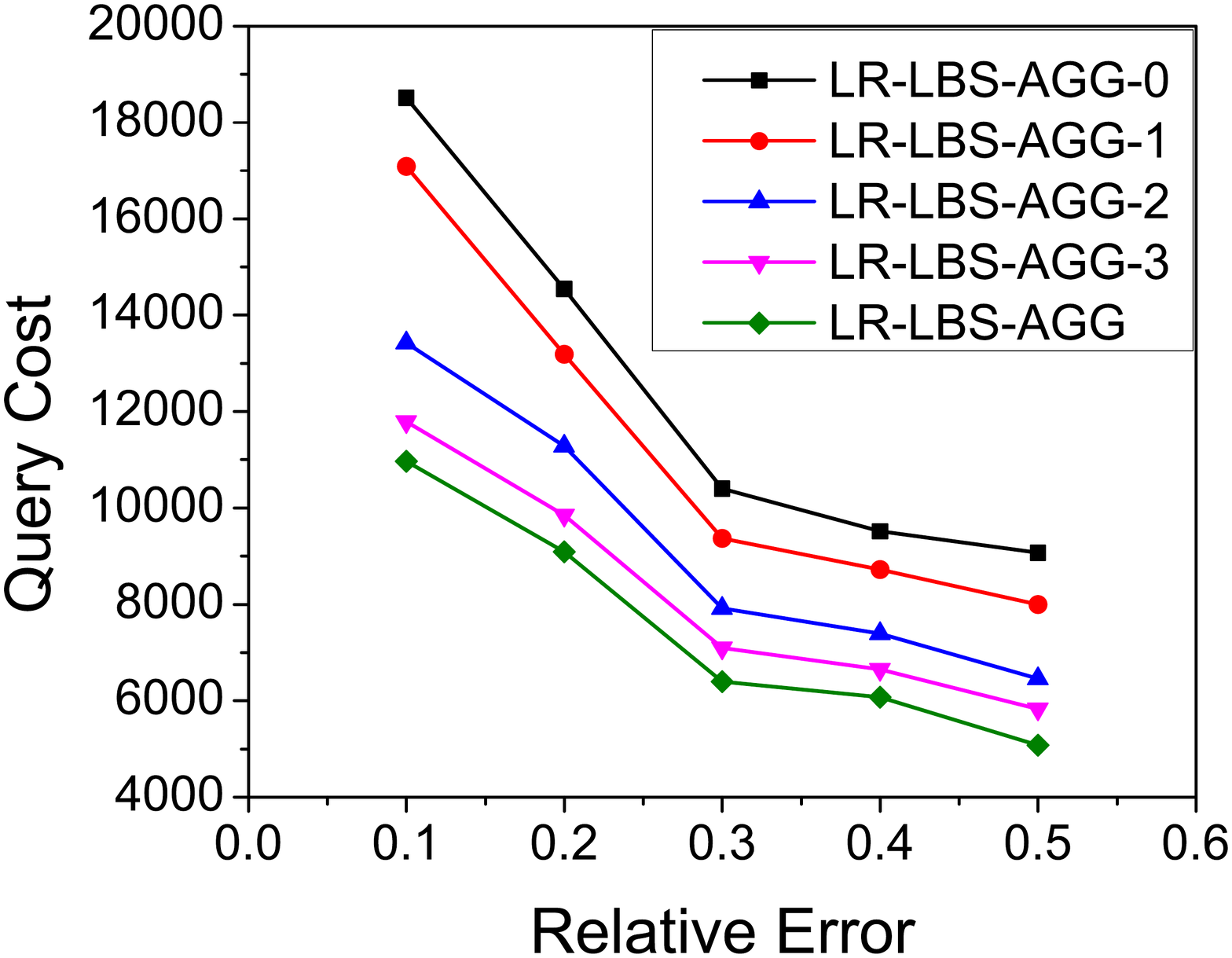}
    \vspace{-7mm}\caption{Query Savings of Error Reduction Strategies}
    \label{fig:ws_seq_opt}
\end{minipage}
\hspace{1mm}
\begin{minipage}[h]{0.48\linewidth}
    \centering
    \includegraphics[width = 45mm, height = 32mm]{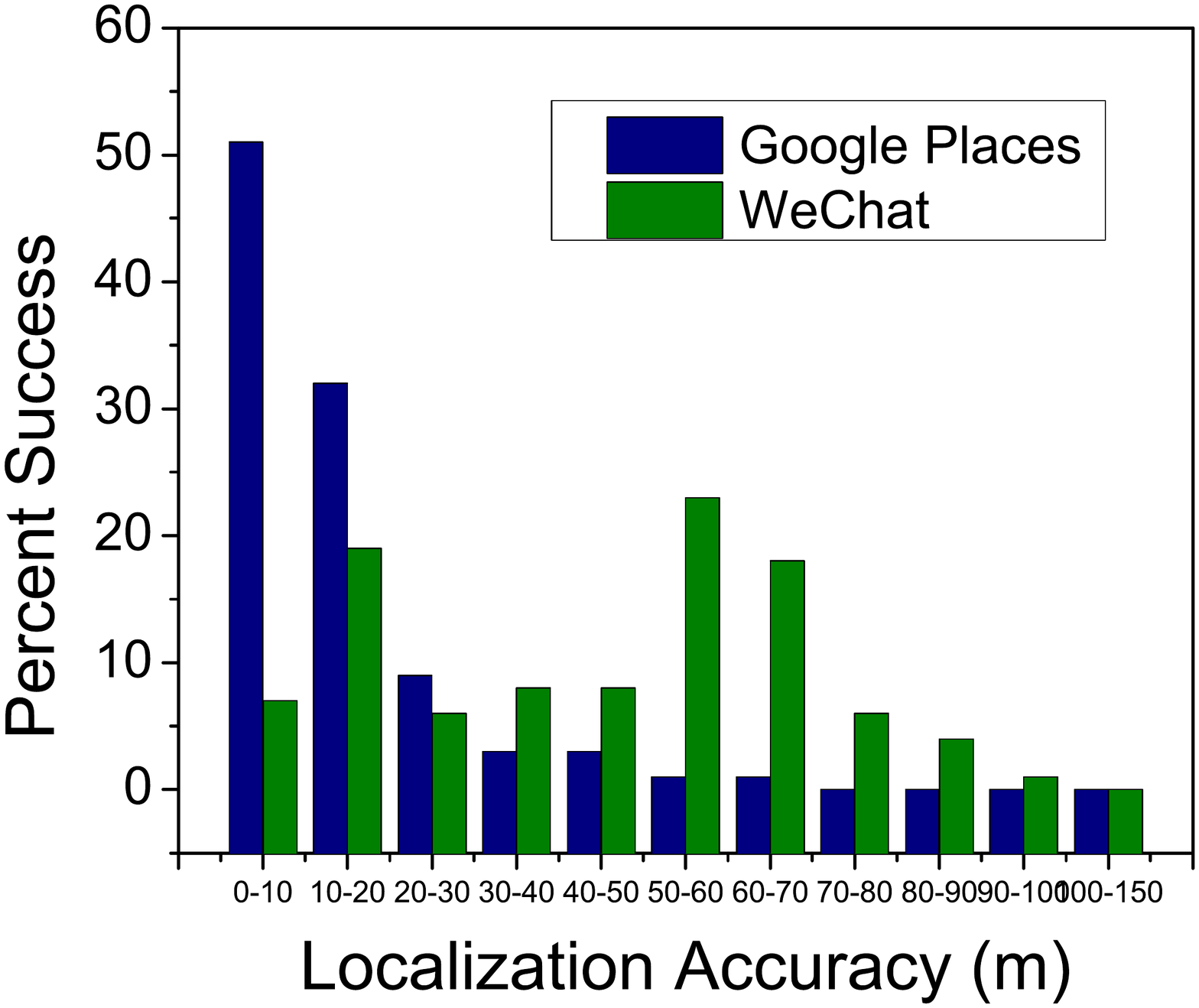}
    \vspace{-7mm}
    \caption{Localization Accuracy}
    \label{fig:tracking_accuracy}
\end{minipage}
\hspace{1mm}
\end{figure}

\vspace{-3mm}
\section{Related Work}\label{sec:relWork}

\noindent{\bf Analytics and Inference over LBS:}
Location based Services (LBS) such as map services (Google Maps) and location based social networks
(such as FourSquare, WeChat, Sina Weibo) are becoming popular in recent years.
The prior work on analytics over LBS focussed exclusively on the LR-LBS scenario.
The closest prior work is \cite{DKA+11} that seeks to estimate COUNT and SUM aggregates over LR-LBS using a nearest neighbor oracle.
It then corrects the bias by using the area of the Voronoi cell using an approach that is very expensive.
Aggregate estimation over LBS such as FourSquare that does not provide nearest neighbor oracle interface could be done using \cite{WHL14,LSW+12}.
\cite{LSW+12} proposed a random region sampling method with an unknown estimation bias that could be eliminated using techniques from \cite{WHL14}.
However, none of them work for LNR-LBS.
There has been work on inferring the location and other private information of users of LBS.
\cite{LZG+13} proposed trilateration based methods to infer the location of users even when the LBS only provided relative distances.
There has been other extensive work\cite{de2013unique,ma2013privacy, srivatsa2012deanonymizing, zang2011anonymization}
on inferring location information and re-identification of users although none of them are applicable for the LBS models studied in this paper.

\noindent{\bf Aggregate Estimations over Hidden Web Repositories:}
There has been a number of prior work in performing aggregate estimation over static hidden databases.
\cite{DJJ+10} provided an unbiased estimator for COUNT and SUM aggregates for {\em static} databases with form based interfaces.
\cite{DDM:07,DZD:09,DZD:10,LTZD:14} describe efficient techniques to obtain random samples from hidden web databases
that can then be utilized to perform aggregate estimation.
Recent works such as \cite{LWA12, DBLP:conf/edbt/WangA11} propose more sophisticated sampling techniques so as to reduce the variance of the aggregate estimation.
For hidden databases with keyword interfaces, prior work have studied estimating the size of
search engines \cite{BG:06,Zhang:2011:MSE:1989323.1989406,zhang2013mining} or a corpus \cite{BFJ+:06, SZS+:06}.

\section{Final Remarks}
In this paper, we explore the problem of aggregate estimation over location based services that are increasingly popular.
We introduced a taxonomy of LBS with $k$-NN query interface based on whether location of the tuple is returned (LR-LBS) or not (LNR-LBS).
For the former, we proposed an efficient algorithm and various error reduction strategies that outperforms prior work.
We initiate study into the latter by proposing effective algorithms for aggregation and inferring the position of tuple
to arbitrary precision which might be of independent interest.
We verified the effectiveness of our algorithms by using a comprehensive set of experiments on a large real-world geographic dataset
and online demonstrations on high-profile real-world websites.

\bibliographystyle{abbrv}
\bibliography{mapco}

\appendix

\section{Binary Search Process}\label{sec:appendixBSearch}

\noindent{\bf Design of Binary Search:} Given the half-line $\ell$ from $c_1$ passing through $c_2$, we conduct the binary search as follows. First, we find $c_\mathrm{b}$, the intersection of this half-line with the bounding box. Then, we perform a binary search between $c_1$ and $c_\mathrm{b}$ to find a segment of the half-line with length at most $\delta$, say with two ends being $c_3, c_4$ (with the distance between $c_3$ and $c_4$ at most $\delta$), such that while $c_3$ returns $t$, $c_4$ returns another tuple, say $t^\prime$. This step takes at most $\log(b/\delta)$ queries, where $b$ is the perimeter of the bounding box.

Then, we consider two half-lines $\ell_1$ and $\ell_2$, both of which start from $c_1$ and form an angle of $-\arcsin(\delta^\prime/r)$ and $+\arcsin(\delta^\prime/r)$ with $\ell$, respectively, where $\delta^\prime$ is a pre-determined (small) threshold and $r$ is the distance between $c_1$ and $c_4$. For each $\ell_i$, we perform the above binary search process to find a (at most) $\delta$-long segment that returns $t$ on one end and $t^\prime$ on the other. Note that such a process might fail - e.g., there might no point on $\ell_i$ which returns $t^\prime$. We set two rules to address this situation: First, we terminate the search for $\ell_i$ if we have reached a segment shorter than $\delta$, with one end returning $t$ and the other returning a tuple other than $t^\prime$. Second, we move on to the next step as long as (at least) one of $\ell_1$ and $\ell_2$ gives us a satisfactory $\delta$-long segment. If neither can produce the segment, we terminate the entire process and output the following (estimated) Voronoi edge: the perpendicular bisector of $(c_3, c_4)$.

Now suppose that $\ell_1$ produces a satisfactory segment of at most $\delta$ long. Let this segment be $(c_5, c_6)$. We simply return our (estimated) Voronoi edge as the line that passes through: (1) the midpoint of $(c_3, c_4)$, and (2) the midpoint of $(c_5, c_6)$. One can see that the overall query cost of the binary search process is at most $3\log(b/\delta)$.

Algorithm~\ref{alg:bsearch} provides the pseudocode for Binary Search process.
\begin{algorithm}[!h]
\caption{{\bf Binary-Search}}
\begin{algorithmic}[1]
\label{alg:bsearch}
\STATE {\bf Input:} Tuple $t$, Locations $c_1, c_2$ where query($c_1$) returns $t$
\STATE {\bf Output:} An edge of $V(t)$ 
\STATE $c_b$ = Intersection of half-line $c_1$, $c_2$ with bounding box
\STATE Find $c_3, c_4$ s.t. $dist(c_3,c_4)<\delta$ and $query(c_3) \neq query(c_4)$
\STATE $r=dist(c_1, c_4)$ 
\STATE Construct lines $\ell_1, \ell_2$ from $c_1$ with angles $\pm \arcsin(\delta'/r)$
\STATE $(c_5, c_6)$ = line segment on $\ell_1$ or $\ell_2$ with $dist(c_5, c_6) < \delta$ and query$(c_5) \neq$ query$(c_6)$
\STATE {\bf if} none exists, return perpendicular bisector of $(c_3, c_4)$
\STATE {\bf else} return line segment passing through midpoints of $(c_3, c_4)$ and $(c_5, c_6)$
\end{algorithmic}
\end{algorithm}

\vspace{1mm}
\noindent{\bf Error Bound on Edge Estimation:} We have the following theorem on the error bound of this binary search process:
\begin{theorem} \label{thm:eee}
For a given tuple $t$ and query location $c_1$ which returns $t$, for any other location $c_2$, the Voronoi cell of $t$ must have an edge $\ell_\mathrm{V}$ that intercepts half-line $(c_1, c_2)$ such that the maximum edge error for estimating $\ell_\mathrm{V}$ satisfies
\begin{align}
\epsilon \leq \max (2\delta^\prime, b \cdot \sin(\arctan(\delta/\delta^\prime))).
\end{align}
In other words, for every point $p \in \ell_\mathrm{V}$, there exists a point $p^\prime$ on our estimated Voronoi edge $\ell^\prime_\mathrm{V}$ generated from $(c_1, c_2)$ (i.e., $p^\prime \in \ell^\prime_\mathrm{V}$), such that
\begin{align}
d(p, p^\prime) \leq \max (2\delta^\prime, b \cdot \sin(\arctan(\delta/\delta^\prime))),
\end{align}
where $d(\cdot,\cdot)$ is the Euclidean distance between two points. In addition, for every vertex $v$ of $\ell_\mathrm{V}$, if line segment $(t, v)$ intercepts $\ell^\prime_\mathrm{V}$, then the interception point $v^\prime$ must satisfy
\begin{align}
d(t, v) - d(t, v^\prime) \leq \max (2\delta^\prime, b \cdot \sin(\arctan(\delta/\delta^\prime))). \label{equ:eee}
\end{align}
\end{theorem}

A simple observation from the theorem is that the binary search process can reach an arbitrary precision level - i.e., for any given upper bound on $d(p, p^\prime)$, say $d_\mathrm{U}$, we can set $\delta^\prime = d_\mathrm{U}/2$ and
\begin{align}
\delta \leq \tan\left(\arcsin\left(\frac{d_\mathrm{U}}{b}\right)\right) \cdot \frac{d_\mathrm{U}}{2}
\end{align}
to satisfy the bound. Since both $\tan$ and $\arcsin$ can be bounded from both sides by a polynomial of its input (through Taylor expansion), one can see that the corresponding query complexity is $O(\log(b/d_\mathrm{U}))$, leading to the following corollary on the maximum error edge defined in \S\ref{sec:lrlbsagg}.
\newtheorem{corollary}{Corollary}
\begin{corollary} \label{thm:ee2}
The query cost required for achieving a maximum edge error of $\epsilon$ is $O(\log(b/\epsilon))$ - i.e., $O(\log(1/\epsilon))$ when we consider the bounding box size $b$ to be constant.
\end{corollary}

\vspace{1mm}
\noindent{\bf Error Bound on Voronoi Cell Volume Estimation:} A direct corollary from Theorem~\ref{thm:eee} is an error bound on the estimated volume of a Voronoi cell. Note from our design of LNR-LBS-AGG that our estimated Voronoi cell is always a subregion of the real one. This, in combination with (\ref{equ:eee}) in Theorem~\ref{thm:eee}, leads to the following corollary.

\begin{corollary} \label{thm:ee3}
For a given tuple $t$, the ratio between the volume of the estimated Voronoi cell $V^\prime(t)$ and the real one $V(t)$ satisfies
\begin{align}
\left(\frac{d - \epsilon}{d}\right)^2 \leq \frac{|V^\prime(t)|}{|V(t)|} \leq 1
\end{align}
where $d$ is the nearest distance between $t$ and another tuple in the database, and $\epsilon$ is the maximum edge error.
\end{corollary}

\end{document}